\documentclass[12pt]{amsart}
\usepackage[top=1.25in,bottom=1in,left=1in,right=1in]{geometry}
\usepackage[utf8]{inputenc} 
\usepackage{amsmath}
\usepackage{amssymb}
\usepackage{amsthm}
\usepackage{bbm}
\usepackage{xcolor}
\usepackage{graphicx}
\usepackage{float}
\usepackage{enumerate}
\usepackage{natbib}
\usepackage{setspace,lipsum}
\usepackage{mathtools}
\usepackage{subcaption}
\usepackage{comment}

\usepackage{booktabs} 
\usepackage{siunitx} 

\usepackage{mathrsfs}
\usepackage{comment}
\usepackage{apptools}
\usepackage{url}
\newcommand{\RN}[1]{  \textup{\uppercase\expandafter{\romannumeral#1}}}

\newtheorem{proposition}{Proposition}
\newtheorem{theorem}{Theorem}
\newtheorem{corollary}{Corollary}
\newtheorem{lemma}{Lemma}

\newtheorem{assumption}{Assumption}

\theoremstyle{remark}
\newtheorem{example}{Example}
\theoremstyle{remark}
\newtheorem{remark}{Remark}
\newtheorem{definition}{Definition}

\usepackage{etoolbox}
\makeatletter
\patchcmd{\@maketitle}
  {\ifx\@empty\@dedicatory}
  {\ifx\@empty\@date \else {\vskip3ex \centering\footnotesize\@date\par\vskip1ex}\fi
   \ifx\@empty\@dedicatory}
  {}{}
\patchcmd{\@adminfootnotes}
  {\ifx\@empty\@date\else \@footnotetext{\@setdate}\fi}
  {}{}{}
\makeatother

\makeatother

\AtAppendix{
\numberwithin{equation}{section}
\numberwithin{definition}{section}
\numberwithin{theorem}{section}
\numberwithin{proposition}{section}
\numberwithin{lemma}{section}
\numberwithin{corollary}{section}}

\newcommand{\1}{\mathbbm{1}}
\usepackage[hidelinks]{hyperref}

\usepackage{multicol}

\allowdisplaybreaks

\title[Testing the Fairness-Accuracy Improvability of Algorithms]{Testing the  Fairness-Accuracy Improvability of Algorithms}
\thanks{We thank Denis Chetverikov, Ivan Canay, Joel Horowitz, Yiqi Liu, Francesca Molinari, Matthew Murphy, Azeem Shaikh, Jann Spiess, and Eisho Takatsuji for helpful comments, as well as conference participants at Bravo/SNSF, MLESC,  and ESIF. We acknowledge National Science Foundation Grants SES-1851629, SES-2149408 and SES-2149422 for financial support. We thank Juri Trifonov for excellent research assistance.}
\author[Eric Auerbach]{Eric Auerbach$^\dag$}
\thanks{$^\dag$Department of Economics, Northwestern University}
\author[Annie Liang]{Annie Liang$^\dag$}
\author[Kyohei Okumura]{Kyohei Okumura$^\dag$}
\author[Max Tabord-Meehan]{Max Tabord-Meehan$^\S$}
\thanks{$^\S$Department of Economics, University of Chicago}
\date{\today} 
\begin{document}

\begin{abstract}
Many organizations use algorithms that have a \emph{disparate impact}, i.e., the benefits or harms of the algorithm fall disproportionately on certain social groups. Addressing an algorithm’s disparate impact can be challenging, however,  because it is often unclear whether it is possible to reduce this impact without sacrificing other objectives of the organization, such as accuracy or profit. Establishing the improvability of algorithms with respect to multiple criteria is of both conceptual and practical interest: in many settings, disparate impact that would otherwise be prohibited under US federal law is permissible if it is necessary to achieve a legitimate business interest. The question is how an analyst can formally substantiate, or refute, this ``necessity’’ defense. In this paper, we provide an econometric framework for testing the hypothesis that it is possible to improve on the fairness of an algorithm without compromising on other pre-specified objectives. Our proposed test is simple to implement and can be applied under any exogenous constraint on the algorithm space. We establish the large-sample validity and consistency of our test, and microfound the test's \emph{robustness to manipulation} based on a game between a policymaker and the analyst. Finally, we apply our approach to evaluate a healthcare algorithm originally considered by \citet{Obermeyer2019-te}, and quantify the extent to which the algorithm's disparate impact can be reduced without compromising the accuracy of its predictions.

\end{abstract}
\maketitle

\section{Introduction}

As algorithms are increasingly used to guide important predictions about people (e.g., which patients to treat, or which borrowers to grant a loan to), substantial concern has emerged about their potential \emph{disparate impact}---namely, that the algorithm's benefits or harms may be born unequally across different social groups. Disparate impact has been empirically documented in a range of applications \citep[see for instance][]{ProPublica,Obermeyer2019-te,ArnoldDobbieHull,Pinkham}, but the organizations that deploy these algorithms also value other objectives such as accuracy and profit. When an algorithm has a disparate impact, is it possible to reduce that disparity without compromising the organization's other objectives? The answer to this question is legally relevant  \citep{barocas2016big}. In many settings, a policy with a disparate impact that would otherwise be prohibited under US federal law is permissible if it is necessary to achieve a legitimate business interest (see Section \ref{subsec:law} for a discussion). As a result, it is essential for both external regulators and internal stakeholders to have tools that can determine when fairness is in conflict with other objectives, and when instead there exist alternative algorithms that do better on all of the specified criteria.  In this paper, we provide an econometric framework and results for establishing, or refuting, the existence of such algorithms. 

Our paper is organized as follows. In the first part,  we formally define improvability with respect to a fairness objective and an accuracy objective (where we use ``accuracy'' as an umbrella term for any objective not directly related to inequities across groups). In the second part, we provide a procedure for testing the improvability of a status quo algorithm given data.  In the third part, we establish the theoretical foundations of our test: We prove its large-sample validity and consistency (under appropriate conditions), and additionally demonstrate that it satisfies a form of robustness to manipulation, which we define in the context of a game between a policymaker and an analyst. Finally, we apply our procedure to test the improvability of a healthcare algorithm studied in \citet{Obermeyer2019-te}.

Section \ref{sec:Framework} presents our conceptual  framework. To define improvability, we build on the algorithmic fairness literature covered in Section \ref{subsec:RelatedLit}, and more specifically the framework of \citet{liang2022algorithmic}. An algorithm is defined to be a mapping from a space of observed covariate vectors (e.g., a medical profile) into a prediction (e.g., a medical diagnosis). There are two predefined groups. Algorithms are evaluated in terms of their accuracy for each group and their fairness across groups. Accuracy for each group is defined to be the group's expected utility with respect to a pre-specified utility function, e.g., the algorithm's correct classification rate or the fraction of patients in this group that get treated. Fairness is possibly based on a different utility function (but need not be); it is measured as the disparity between the expected utilities of the two groups. This very general formulation allows us to nest most objectives that have been proposed in the preceding literature on algorithmic fairness, as well as the broad range of business objectives that could be legally relevant (see Section \ref{sec:Setup} for examples).

Our analysis allows for the analyst to pre-specify a class of permissible algorithms. We define an algorithm to be \emph{$(\Delta_r,\Delta_b,\Delta_f)$-improvable} within this class if some other permissible algorithm simultaneously achieves higher accuracy for each group (respectively by factors of $1+\Delta_r$ and $1+\Delta_b$) and lower  unfairness (by a factor of $1-\Delta_f$). This parameterized definition permits us to move beyond simply testing whether any improvement exists, and instead to examine specific magnitudes of potential gains. Such magnitude considerations are particularly important in legal and policy contexts when the practical significance of a disparate impact reduction is taken into consideration \citep{tobia2017disparate,DOJTitleVI2024}.\footnote{In the context of employment discrimination, \citet{tobia2017disparate} reports that ``The First, Third, and Tenth Circuits oppose practical significance inquiries; the Second, Fourth, Fifth, Sixth, Ninth, and Eleventh Circuits endorse them; and the D.C., Seventh, and Eighth Circuits have no clear precedent.'' In the context of discrimination by an organization receiving federal funding, the Department of Justice is less equivocal. In their Title VI legal manual, they write, ``investigating
agencies must determine whether the disparity is large enough to matter, i.e., is it sufficiently
significant to establish a legal violation. The magnitude of the disparity necessary may be
difficult to define in some cases, but guidance can be drawn both from judicial consideration of
this question and from federal agency guidelines'' \citep{DOJTitleVI2024}. See Section 1.1 and Appendix Section D of this paper for a discussion of the legal framework.}  By varying those parameters in our subsequent statistical procedure, we can test whether the data supports, or fails to reject, different sizes of fairness and accuracy-improvements. 

We suppose that there is a status quo algorithm whose disparate impact is of interest, and the analyst would like to test the null hypothesis that this algorithm is not $(\Delta_r,\Delta_b,\Delta_f)$-improvable. The analyst has access to a sample of individuals, their covariates, and a ``ground truth’’ outcome (i.e., a quantity that determines the optimal algorithm and is known ex-post). Section \ref{sec:Approach} describes our proposed procedure: First, the analyst splits the data into a training set and a test set, and uses the training set to identify a candidate algorithm that potentially improves on the status quo algorithm. Subject to certain regularity conditions, any method for identifying a candidate algorithm---which we call a \emph{selection rule}---can be plugged into this step (see examples in Section \ref{subsec:RelatedLit}). Once a candidate algorithm has been identified, the analyst uses the test set to evaluate whether the simultaneous improvements of the candidate algorithm for both criteria are statistically significant. To avoid computing standard errors case-by-case for each specification of accuracy and fairness, we propose using the bootstrap to compute the critical values of our test. 
Finally, to reduce the uncertainty introduced by sample-splitting, we recommend that the analyst repeats the above process several times and aggregates the results by reporting the median $p$-value across the set of tests. The analyst rejects the null hypothesis whenever the median $p$-value falls below half of the desired significance level.

Our proposed approach overcomes a key practical challenge, which is that algorithms are often exogenously constrained for legal, logistical, or ethical reasons. The exact nature of these constraints may vary across settings, and include (among others) capacity constraints on how many individuals are admitted or treated, shape constraints such as monotonicity of the decision with respect to a given input \citep{pmlr-v108-wang20e,Felders2000,NEURIPS2020_b139aeda}, and statistical constraints such as requiring independence of decisions with a group identity variable \citep{Mitchelletal,barocas-hardt-narayanan}. Sample splitting and the bootstrap ensure that our test procedure is valid regardless of the specific class of algorithms, constraints, or utility functions chosen, subject to the regularity conditions given in our stated assumptions.

 Sections \ref{sec:Full_test} and \ref{sec:RobustManipulation} provide theoretical justifications for our procedure. Section \ref{sec:Full_test} establishes that our test is asymptotically valid and---whenever the selection rule satisfies a condition we call \emph{improvement convergence}---it is also consistent. These results build on a bootstrap consistency lemma that we develop to handle two specific features of our setting: First, our test statistic is constructed using absolute values, and it is well-known that bootstrap consistency fails at points of non-differentiability \citep[see e.g.,][]{fang2019inference}; second, the fact that we impose very few restrictions on the selection rule means that the distribution of our test statistics may vary with the sample size and are not guaranteed to ``settle down" in the limit. Since (to the best of our knowledge) generally available bootstrap consistency results do not cover our exact setting, we develop a new result building on the work of \cite{mammen2012does} (see Appendix \ref{sec:boot_consist} for details).

Section \ref{sec:RobustManipulation} provides a game-theoretic framework to explain how our use of repeated sample-splitting leads to a test that is more robust to manipulation. We study a game between an analyst and a policymaker, where the policymaker chooses the statistical procedure, and the analyst chooses a number of times to repeat this procedure. The analyst then selects one $p$-value to report, which determines if the null hypothesis is rejected. We suppose that the analyst would like to reject the null even when it holds, and define a test to be more robust to manipulation than another if it leads to a lower probability of (incorrect) rejection under the null. We prove that aggregating $p$-values is indeed more robust to manipulation than using a single split. Although this microfoundation is particularly relevant in disparate impact testing, where the analyst may be incentivized to reject the null hypothesis even when it is true, our formal results hold for any valid test and thus apply more broadly.

Finally, to illustrate our proposed procedure, we revisit the data and setting of \citet{Obermeyer2019-te}. The status quo algorithm is one used by a hospital to identify patients for automatic enrollment into a high-intensity care management program. Using our framework and test procedure, we reject the null hypothesis that it is not possible to simultaneously improve on the accuracy and the fairness of the algorithm, and we moreover quantify the extent of possible improvement. We find that large improvements in fairness are possible without compromising accuracy, while we are unable to establish that the reverse is true.

\subsection{Legal framework} \label{subsec:law}
The econometric framework we develop in this paper is broadly applicable whenever a policymaker wants to examine whether it is possible to increase the fairness of a policy without sacrificing other objectives. However, we view our contribution as particularly relevant to regulators and other legal decision makers tasked with evaluating cases of alleged disparate impact under US federal law (see e.g., \citet{barocas2016big},  and  \cite{gillis2019big}, and \citet{yang2020equal}).\footnote{We focus on restrictions of private individuals (as opposed to government agents) from certain policies. As a general rule, discrimination is legally permissible by private actors except when specifically prohibited by law.}  To keep this paper self-contained, we provide an independent summary of the relevant points here. We additionally contribute a brief review of Title VI of the Civil Rights Act of 1964, which is relevant for programs or activities that receive federal funding and covers our empirical application in Section \ref{sec:Application}.

Although there is no singular discrimination law, a body of legislation, regulation, and court precedent prohibits discrimination against certain social groups (such as those defined by race, gender, or religion) called ``protected classes'' in settings as varied as employment, healthcare, education, and lending. US federal law generally distinguishes between two types of discrimination: disparate treatment and disparate impact. A policy has a disparate treatment if it intentionally offers different services based on an individual’s membership in a protected class. It has a disparate impact if one protected class is benefited or harmed more than another by the policy. Our paper pertains to disparate impact, which we broadly review first. Details and specific references to the US Code, the Code of Federal Regulations, and the Department of Justice's Guidelines are deferred to Appendix D. 

A prototypical process for evaluating a disparate impact case has three parts. The first part is to determine whether the benefits or harms disproportionately affect members of a protected class. This is often called a ``prima facie’’ case of discrimination. For this part of the process, there already exist well-established statistical frameworks \citep[see for instance][]{hazelwood1977, shoben1978differential, groves1991}, so we do not focus on it in our paper. 

If the prima facie case is not successful, then the process ends and concludes that there is no disparate impact. If it is successful, then the second part of the process is to determine whether the challenged policy is a business necessity, i.e.,  necessary to achieve some legitimate nondiscriminatory interest. If the business necessity defense is not successful, then the process ends and concludes that the disparate impact is not permissible. If it is successful, then the third and final part of the process is to determine whether there exists a valid alternative that is less discriminatory, but still serves the same interests as the challenged policy. The policy is not permissible if such an alternative exists. Otherwise, it is. Our econometric framework is designed to assess this final part. 

The exact implementation of this three-part process depends on the setting, and in particular on the specific statute, regulation, or court precedent prohibiting discrimination that is allegedly violated. For example, in the context of employment discrimination (which is prohibited by Title VII of the Civil Rights Act of 1964),  the Equal Employment Opportunity Commission and/or private individuals can bring a lawsuit to court to remedy a disparate impact violation. In this setting, the three-part test is adversarial; that is,  there are two opposing parties, where one party is responsible for the evidence and arguments supporting the policy, and the other party is responsible for the evidence and arguments against the policy. In contrast, in the context of discrimination by a program or activity that receives federal funding (which is prohibited by Title VI of the Civil Rights Act of 1964), investigation and enforcement instead falls on the funding agency \citep[see][]{alexander2001}. Though the funding agency’s determination of disparate impact is not typically viewed as adversarial, the Department of Justice still recommends using the three-part test outlined above.\footnote{The Department of Justice is tasked with coordinating the implementation and enforcement of Title VI across federal agencies, by Executive Order
12250, 28 C.F.R. pt. 41, app. A. Their guidelines can be found in the Title VI Legal Manual available at \url{https://www.justice.gov/crt/fcs/T6manual}.}

Twenty-six funding agencies currently enforce Title VI in a variety of contexts such as education, healthcare, and transportation. In Section \ref{sec:Application}, we apply our framework to study a potential disparate impact in access to hospital services. If the hospital receives federal funding from Medicare or Titles VI or XVI of the Public Health Service Act, this discrimination would be prohibited by Title VI of the Civil Rights Act of 1964. In this case, a likely enforcing agency would be the US Department of Health and Human Services (HHS).\footnote{Specific referrals of individual complaints to HHS can be found at \url{https://www.hhs.gov/civil-rights/for-individuals/index.html} and \url{https://www.cms.gov/about-cms/web-policies-important-links/accessibility-nondiscrimination-disabilities-notice}.} As hospitals and other healthcare providers increasingly use algorithms to allocate resources and determine patient care, we view our framework as one that could be used by HHS in order to evaluate allegations of discrimination under Title VI. 

\subsection{Related Literature} \label{subsec:RelatedLit}

Our paper builds on the algorithmic fairness literature: see \cite{chouldechova2018frontiers}, \cite{cowgill2020algorithmic}, or \citet{barocas-hardt-narayanan} for recent overviews.\footnote{Much of this literature is concerned with fair algorithm design. These papers typically build an explicit fairness constraint into the optimization problem, and ask how to find (or approximate) the optimum in a computationally efficient manner (e.g., \citet{dwork2012fairness}, \citet{agarwal2018reductions}, \citet{kusner2019making}, \citet{dwork2018decoupled}).} In particular, our paper is related to those papers that evaluate whether an algorithm with disparate impact can be justified by business necessity. For example, \cite{coston2021characterizing} formulate an optimization problem to find the smallest attainable level of disparate impact over a given set of algorithms that satisfy an additional constraint on business necessity; \cite{viviano2023fair} propose a policy learning procedure to select the fairest policy subject to a constraint on Pareto optimality of the resulting policy; \citet{BlattnerSpiess2023} use a sample-splitting approach (as we do) to find potential fairness and interpretability improvements in lending algorithms; and \cite{gillis2024operationalizing} formulate a mixed integer optimization problem to find the least discriminatory linear classification model, with applications in consumer finance. 

These papers focus on the computational, algorithmic and statistical aspects of how to \emph{ identify} an algorithm that improves on fairness, a problem that we abstract away from in our main text.\footnote{Appendix \ref{sec:Search} presents some novel formulations of optimization problems designed to find less discriminatory alternatives.} Indeed, any of the methods developed in this literature could be used within our proposed procedure as a means to search for candidate algorithms for fairness-improvement. Our focus is instead on the complementary problem of statistical inference. We view our main contribution as  a general purpose blueprint, which allows an analyst to evaluate not just whether an alternative algorithm achieves a better performance on test data, but also whether the joint improvement on fairness and accuracy is statistically meaningful.

To do this, we formalize the property of fairness-improvability as a null hypothesis. We propose a test for this null, and show that (under suitable conditions) it is valid and consistent. Our test involves sample-splitting, an approach with a long history in statistics going back at least to the work of \cite{moran1973dividing} and \cite{cox1975note}. To mitigate the uncertainty introduced by sample splitting, we further recommend that the analyst performs repeated sample-splitting (as in \citet{meinshausen2009p}, \citet{diciccio2020exact}, \citet{wasserman2020universal}, and \citet{ritzwoller2023reproducible}, among others).

In Section \ref{sec:RobustManipulation}, we provide a theoretical argument that repeated sample-splitting makes our test more robust to potential manipulation by the analyst. This section is related to papers that  draw conclusions about optimal research procedures based on models of a strategic researcher \citep{andrews2019identification,FrankelKasy,AndrewsShapiro2021,KitagawaVu2023,Spiess2024}. In particular, our focus on selective reporting of a $p$-value is similar to the concern with $p$-hacking in \citet{jagadeesan2024publication} and \citet{KasySpiess2024}. Different from these papers, we focus specifically on the question of whether manipulation is reduced when $p$-values are aggregated over repeated sample splits. That is, our model does not broadly study the optimal statistical procedure, but rather provides a theoretical justification for the common use of repeated sample-splitting, which is not always justified by more standard econometric properties (e.g., power). 

Finally, a closely related paper is the work of \cite{liu2024inference}, which directly estimates the fairness-accuracy frontier developed by \cite{liang2022algorithmic}. A full characterization of the frontier allows the analyst to answer a broad range of questions related to algorithmic fairness. In particular, once the fairness-accuracy frontier has been characterized,  testing for the fairness-improvability of a status-quo algorithm amounts to testing whether or not the algorithm belongs to the frontier. While our approach focuses on the narrower question of testing for fairness-improvability (rather than characterizing the full frontier), it is more flexible in two ways---(1) it allows for definitions of accuracy and fairness with respect to different utility functions, and (2) it can automatically accommodate (almost) any exogenous constraint on the algorithm space.\footnote{The characterization of the frontier in  \cite{liu2024inference} is developed for the class of all possible algorithms, and (to the best of our knowledge) holds only under specific restrictions on the algorithm space, which rule out global constraints such as capacity constraints and monotonicity constraints.} Both of these features are often relevant for evaluating algorithms that are used in practice. 

\section{Model} \label{sec:Framework}

\subsection{Setup} \label{sec:Setup} Consider a population of individuals, where each individual is described by a covariate vector $X$ taking values in the set $\mathcal{X} \subseteq \mathbb{R}^d$, a type $Y$ taking values in $\mathcal{Y} \subseteq \mathbb{R}$, and a group identity $G$ taking values in $\mathcal{G} =  \{r,b\}$. The covariates in $X$ are the information that the algorithm can access to make a decision, the type $Y$ is a payoff-relevant unknown that the algorithm cannot directly access, and the group identity $G$ is a special covariate (which the algorithm may or may not access\footnote{The algorithm is given only $X$ as input, but we allow for the possibility that $G$ is an element of the vector $X$, or is perfectly correlated with covariates in $X$.}) that is used to evaluate the fairness of the algorithm. We use $P$ to denote the joint distribution of $(X,Y,G)$ across individuals. An \emph{algorithm}  $a : \mathcal{X} \rightarrow \mathcal{D}$ maps covariate vectors into $\mathcal{D}$, the set of possible decisions.\footnote{The map $a$ is often referred to as an assignment rule or policy in econometrics (where the objective of $a$ is usually to make treatment assignment decisions), and as a classifier or prediction rule in machine learning (where the objective of $a$ is typically to predict an unknown outcome). We use the more general term ``algorithm'' to encompass these possibilities and to evoke our motivating applications of algorithmic fairness.} 

We evaluate algorithms from the perspectives of accuracy and fairness, which are defined with respect to an \emph{accuracy utility function}
$u_A: \mathcal{X} \times \mathcal{Y} \times \mathcal{D} \rightarrow \mathbb{R}_+$,  and a (possibly identical)  \emph{fairness utility function}
$u_F: \mathcal{X} \times \mathcal{Y} \times \mathcal{D} \rightarrow \mathbb{R}_+$. To ease exposition, we assume both utility functions are positive-valued.\footnote{This is inessential for our results, but permits an easier statement of Definition \ref{def:DeltaImprove}, where we consider ratios of utilities.} We will consider accuracy and fairness criteria that can be formulated as
\begin{equation} \label{eq:UAUF}U_A^g(a)=\frac{E_{P}[u_A(X,Y,a(X)) \mid G=g]}{E_{P}[w_A(X,Y,a(X))\mid G=g]} \quad \mbox{and} \quad U_F^g(a)=\frac{E_P[u_F(X,Y,a(X)) \mid G=g]}{E_{P}[w_F(X,Y,a(X))\mid G=g]}
\end{equation}
where $w_A$ and $w_F$ are (again positive-valued) normalization functions (see Examples \ref{ex:Classification}-\ref{ex:Profit} for possible specifications). Thus, $U_A^g$ and $U_F^g$  represent normalized expected utilities for group $g$ with respect to the accuracy and fairness utility functions, respectively. When the utility functions and normalization functions are common across accuracy and fairness, we drop the subscripts and simply write $U^g(a)$.

We say that one algorithm is more accurate than another if the value of $U_A^g$ is larger for both groups, and more fair if the value of $U^g_F$ differs less across groups. 
\begin{definition} \label{def:MoreAccurateFair} Algorithm $a_1$ is \emph{more accurate} than algorithm $a_0$ if \[U_A^r(a_1) >  U_A^r(a_0) \quad \mbox{and} \quad U_A^{b}(a_1) > U_A^{b}(a_0)\]
and 
 \emph{more fair} than algorithm $a_0$  if 
 \[ \vert U_F^r(a_1) - U_F^{b}(a_1) \vert < \vert U_F^r(a_0) - U_F^{b}(a_0) \vert.\]
 \end{definition}

We make some remarks on the formulations of accuracy and fairness below.

\begin{remark}[Alternative definitions] In some applications, the analyst may wish to consider a single accuracy measure (e.g., the expected utility in the entire population\footnote{While this is a popular  aggregated measure of accuracy (see e.g., \citet{KleinbergMullainathan}), depending on the application the analyst may prefer other weights on group utilities. For example, if the algorithm is used to make lending decisions and one group is wealthier than the other, then loans for the wealthier group may be of higher value, so that incorrectly classifying creditworthy individuals in this group is more costly. And from a welfare perspective, generalized utilitarian criteria of the form $\alpha_r U_A^r(a) + \alpha_b U_A^b(a)$ (with weights $\alpha_r,\alpha_b>0$ different from population proportions) are widely used \citep{SaezStantcheva16,CharnessRabin02,DworczakKominersAkbarpour21}. We thus focus on our more stringent test of accuracy-improvement in Definition \ref{def:MoreAccurateFair}, which requires improvements for both groups.}), or to require simultaneous fairness improvements according to several different utility functions.\footnote{For example, the analyst may require a fairness improvement not only with respect to the utility function given in Example \ref{ex:Calibration}, but also with respect to the mirrored utility function that conditions on $D=0$. Indeed, the property popularly known as calibration requires improvements according to both of these measures \citep{chouldechova,KMR}.} Our procedure and results can be straightforwardly extended in any of these cases, but we present our results for the three criteria $U_A^r$, $U_A^b$, and $\vert U_F^r-U_F^b\vert$ proposed here.

\end{remark}

\begin{remark}[Two-sided fairness measure] Our fairness measure is ``two-sided'' in that it is the absolute value of the difference in expected utilities for the two groups. This approach makes sense when discrimination is prohibited regardless of which group experiences harm.\footnote{For example, 
Title VII of the Civil Rights Act prohibits workplace discrimination for any group based on race, color, religion, sex, or national origin. The
Equal Credit Opportunity Act (ECOA) similarly requires fair lending practices for all demographic groups.
}$^,$\footnote{We view the two-sided measure as especially appropriate for our setting since algorithmic interventions may inadvertently reverse existing disparities.} In some  contexts, regulations  focus on protecting specific historically disadvantaged groups, and the analyst may modify the fairness measure to consider only scenarios where the protected group receives lower utility than the unprotected group (i.e., using the positive part function rather than the absolute value function). Our framework and results can be adapted to this case.
\end{remark}

\begin{remark}[Implications for social welfare] While Definition \ref{def:MoreAccurateFair} is motivated by legal constraints on algorithmic discrimination, it also aligns with welfare improvements under diverse social welfare criteria. These  include the utilitarian criterion $P(G=r) \cdot U_A^r(a) + P(G=b) \cdot U_A(a)$, the Rawlsian criterion $\min\{U_A^r(a), U_A^b(a)\}$, and any positive linear combination of fairness and accuracy $\alpha_r U_A^r(a) + \alpha_b U_A^b(a) + \alpha_f \vert U_F^r(a)-U_F^b(a) \vert$, among many others (see \citet{liang2022algorithmic} for details).\footnote{See \citet{LiangLu} however for an argument that this framework (and other popular frameworks in the algorithmic fairness literature) are not compatible with social welfare approaches based on ``veil of ignorance'' thought experiments.} Specifically, if one algorithm is simultaneously more fair and more accurate than another, then a social planner with any of these welfare criteria will prefer the former algorithm.
\end{remark}

How $U^g_A$ and $U^g_F$ are defined is application-specific, and we list below example utility specifications from the prior literature (see e.g. \citet{Mitchelletal}), dropping the fairness and accuracy subscripts to simplify notation. Although we refer to $U^g_A$ and $U^g_F$ as (expected) ``utility'' throughout, in many of these examples, what it quantifies has a unit, for example a probability of error or a dollar profit.\footnote{For this reason we do not view invariance to affine transformations as a necessary (or in some cases, desirable) property of our subsequent tests.}

\begin{example}[Classification Rate] The decision is a prediction of $Y$, and both are binary. Let $u(x,y,d) = \mathbbm{1}(y=d)$ be an indicator variable for whether the prediction is correct, and let the normalization function be $w(x,y,d)=1$. Then
\[
U^g(a) = P(Y=a(X) \mid G=g)
\]
is the correct classification rate of algorithm $a$ within group $g$, while  $\vert U^r(a)-U^{b}(a)\vert$ is the disparity in classification rates across groups. \label{ex:Classification}
\end{example}

\begin{example}[Calibration]
The type $Y$ takes values in $\mathbb{R}$ and the decision is binary.
Define $u(x,y,d) = y\mathbbm{1}(d=1)$ and $w(x,y,d)=\mathbbm{1}(d=1)$.
Then 
\[
U^g(a)= E[Y\mid a(X)=1,G=g]\]
is the expected type for individuals receiving decision $a(X)=1$ in group $g$, while $\vert U^r(a)-U^{b}(a)\vert$ evaluates how different these conditional expectations are across groups.
\label{ex:Calibration}
\end{example}

\begin{example}[False Positive Rate] \label{ex:FalsePositive} Both $Y$ and the decision are binary. Define 
$u(x,y,d) = \mathbbm{1}(y=0,d=1)$ and $w(x,y,d)=\mathbbm{1}(y=0)$.
Then 
\[U^g(a)=P(a(X)=1 \mid Y=0,G=g)\]
is the probability of choosing the (wrong) decision $a(X)=1$ for individuals with true type $Y=0$ in  group $g$, while $\vert U^r(a)-U^{b}(a)\vert$ compares the probability of this  error across groups.
\end{example}

\begin{example}[Profit] \label{ex:Profit} There is a good that is offered to an individual at price $d$. The individual's unknown willingness-to-pay is $Y$. The utility function is \[u(x,y,d) = \left\{\begin{array}{cc}
d & \mbox{ if } y \geq d \\
0 & \mbox{ otherwise} \end{array} \right.\]
and the normalization function is $w(x,y,d)=1$. That is, the firm receives the price if the individual's willingness to pay exceeds that price, and otherwise the firm receives nothing.
Then $U^g(a)$ is the average profit the firm receives from group $g$ consumers, and $\vert U^r(a)-U^{b}(a)\vert$ is the difference in average profit between consumers in the two groups.
\end{example}

For a fixed class of algorithms $\mathcal{A}$, the feasible accuracy and unfairness levels are those $(U_A^r, U_A^{b}, \vert U_F^r - U_F^{b}\vert)$ triples that can be achieved by some algorithm in $\mathcal{A}$. Our proposed testing procedure will be agnostic about the specific choice of $\mathcal{A}$, and in particular we will allow for $\mathcal{A}$ to be exogenously restricted in some way. We view this flexibility as important since (as discussed in the introduction) in practice there are often legal, logistical, or ethical constraints on which algorithms can be implemented.

\subsection{Accuracy/Fairness Improvability}\label{sec:af_improve}

There is a \emph{status quo algorithm} $a_0: \mathcal{X} \rightarrow \mathcal D$ representing an algorithm that is already in use, or which has been proposed for use. The status quo algorithm $a_0$ has some level of disparate impact, $|U_F^r(a_0)- U_F^{b}(a_0)|$, which we assume throughout the paper is strictly positive.

\begin{assumption}\label{assp:n} The status quo algorithm is not perfectly fair, i.e., $|U_F^r(a_0)- U_F^{b}(a_0)| > 0$. 
\end{assumption}

This assumption reflects the basic motivation of our paper: If the status quo algorithm has no disparate impact (corresponding to $|U_F^r(a_0)- U_F^{b}(a_0)| = 0$), then there is no reason to look for fairness-improving alternatives. Supposing that Assumption \ref{assp:n} holds, we ask whether it is possible to improve upon the disparate impact of the status quo algorithm without compromising on its accuracy. Formally, we extend Definition \ref{def:MoreAccurateFair} in the following way to quantify the extent of fairness- or accuracy-improvement.

\begin{definition}[Accuracy-Fairness Improvement]\label{def:DeltaImprove} Fix any $\Delta_r,\Delta_{b},\Delta_f \in \mathbb{R}$. The algorithm $a_1$ \emph{$(\Delta_r,\Delta_{b},\Delta_f)$-improves} on the algorithm $a_0$ if
\[\frac{U_A^r(a_1)}{U_A^r(a_0)} > 1 + \Delta_r, \quad \frac{U_A^{b}(a_1)}{U_A^{b}(a_0)} > 1 + \Delta_{b}, \quad \mbox{and} \quad \frac{\vert U_F^r(a_1) - U_F^{b}(a_1)\vert}{\vert U_F^r(a_0) - U_F^{b}(a_0)\vert} < 1 - \Delta_f~.\]
\end{definition}

That is, one algorithm $(\Delta_r, \Delta_b, \Delta_f)$-improves on another if its accuracy for group $r$ is at least $(1+\Delta_r)$ times as high, its accuracy for group $b$ is at least $(1+\Delta_b)$ times as high, and its disparate impact is no more than $(1-\Delta_f)$ times as high.\footnote{Analogous results to those presented in Sections \ref{sec:Approach} and \ref{sec:Full_test} hold under a modification of Definition \ref{def:DeltaImprove}, where we replace the ratios in this definition with differences.}

The special case of $\Delta_r=\Delta_b=\Delta_f=0$ corresponds to whether disparate impact can be reduced at all without compromising on accuracy, and we name it  \emph{FA-dominance}.

\begin{definition}[FA-dominance]
Fix a class of algorithms $\mathcal{A}$. The algorithm $a_0$ is \emph{FA-dominated within class $\mathcal{A}$} if there exists an algorithm $a_1 \in \mathcal{A}$ that $(0,0,0)$-improves on $a_0$, i.e.,
\[\frac{U_A^r(a_1)}{U_A^r(a_0)} > 1, \quad \frac{U_A^{b}(a_1)}{U_A^{b}(a_0)} > 1 , \quad \mbox{and} \quad \frac{\vert U_F^r(a_1) - U_F^{b}(a_1)\vert}{\vert U_F^r(a_0) - U_F^{b}(a_0)\vert} < 1~.\]
\end{definition}

Some courts consider only whether discrimination can be reduced without compromising business necessity, while others examine the magnitude of potential reductions (see \citet{tobia2017disparate} for a discussion of this topic in the context of Title VII workplace discrimination). The importance of magnitude is particularly emphasized in Title VI cases, where the Department of Justice's Title VI Legal Manual explicitly recommends consideration of whether the disparity is ``large enough to matter'' \citep{DOJTitleVI2024}. Our more general Definition \ref{def:DeltaImprove} provides a framework that addresses both approaches: its use in our subsequent test allows us to systematically evaluate, based on observed data, both the existence and extent of feasible improvements.

\section{Proposed Approach} \label{sec:Approach}

We consider an analyst who does not know the joint distribution of $(X,Y,G)$, but instead observes a sample $\{(X_i,Y_i,G_i)\}_{i=1}^n$ consisting of $n$ independent and identically distributed observations from this distribution. The analyst's objective is to test $(\Delta_r,\Delta_b,\Delta_f)$-improvability of a status quo algorithm $a_0$ given such a sample. Formally, the analyst  specifies a set of algorithms $\mathcal{A}$ and tests the null hypothesis
\begin{equation}\label{eq:H0_full}
H_0: \text{Algorithm $a_0$ is not $(\Delta_r,\Delta_b,\Delta_f)$-improvable within class $\mathcal{A}$},
\end{equation}
against the alternative that such an improvement exists. These parameters should be set to whatever values are relevant for the application (or varied over a range of relevant values, as in our subsequent application in Section \ref{sec:Application}); in a test of FA-dominance, the analyst chooses $(\Delta_r,\Delta_b,\Delta_f)=(0,0,0)$.

Section \ref{sec:Description} outlines the approach, Section \ref{sec:Test} provides details, and the subsequent Section \ref{sec:Full_test} establishes the validity and consistency of our procedure under stated assumptions.

\subsection{Description of Procedure} \label{sec:Description}

The analyst first chooses a \emph{selection rule} that maps  samples into a choice of algorithm from $\mathcal{A}$, i.e., a mapping
 \[\rho: \mathcal{S} \rightarrow \mathcal{A}\]
 where $\mathcal{S} = \bigcup_{m \geq 1} \mathcal{S}_m \equiv \bigcup_{ m \geq 1} (\mathcal{X} \times \mathcal{Y} \times \mathcal{G})^m$ is the set of all finite samples of observations. Our results apply for any such rule subject to certain regularity conditions which we introduce in Section \ref{sec:Test}.

 \smallskip

 In each round $k = 1, \dots, K$  of the procedure (where $K$ is a user-specified parameter), the analyst repeats the following steps:

\smallskip

1. \emph{Split the sample into train and test.} The training sample $S_{train}^n$ includes $m_n = \lfloor \beta \cdot n \rfloor$ observations selected uniformly at random from $\{(X_i, Y_i,G_i)\}_{i=1}^n$, where $\beta \in (0,1)$ is a user-chosen fraction (we use $\beta=1/2$ in our application). The test sample $S_{test}^n$ consists of the remaining $\ell_n = n - m_n$ observations. 

\smallskip

2. \emph{Find a candidate algorithm using the training sample.} Apply the selection rule $\rho$ to the training sample $S_{train}^n$. To simplify notation, we use $\hat{a}^{\rho}_{1n}$ to denote the (realized) algorithm $\rho(S_{train}^n)$, which serves as a candidate algorithm in the subsequent step.

\smallskip

3. \emph{Test whether $\hat{a}^{\rho}_{1n}$ constitutes a $(\Delta_r,\Delta_b,\Delta_f)$-improvement relative to $a_0$.} Given the status quo algorithm $a_0$ and the candidate algorithm $\hat{a}^{\rho}_{1n}$ obtained in the previous step, test the null hypothesis
\begin{equation} \label{null:Fair}
H_0: \text{$\hat{a}^{\rho}_{1n}$ does not constitute a $(\Delta_r,\Delta_b,\Delta_f)$-improvement on $a_0$}
\end{equation}
using the test outlined in Section \ref{sec:Test} and the sample $S_{test}^n$. 
Let $p_k$ denote the $p$-value associated with this test (details below in Section \ref{sec:Test}).

\smallskip

4. Repeat steps 1-3 $K$ times, yielding a vector of $p$-values $(p_1, \dots, p_K)$.

\smallskip

5. Define $p_{\mathrm{med}}$ to be the median\footnote{An inspection of the proof of Corollary \ref{cor:split_test} shows that the result holds for any standard definition of the median, but the formal definition of the corresponding test $\psi_n$ below fixes $p_{\rm med}$ to be the smallest such choice.} of the values $(p_1, \dots, p_K)$ and reject $H_0$ if and only if $p_{\mathrm{med}} < \alpha/2$.
Our Corollary~\ref{cor:split_test} guarantees that under certain regularity conditions, this test is of level $\alpha$. 

\begin{figure}[h] 
\begin{center}
    \includegraphics[scale=0.33]{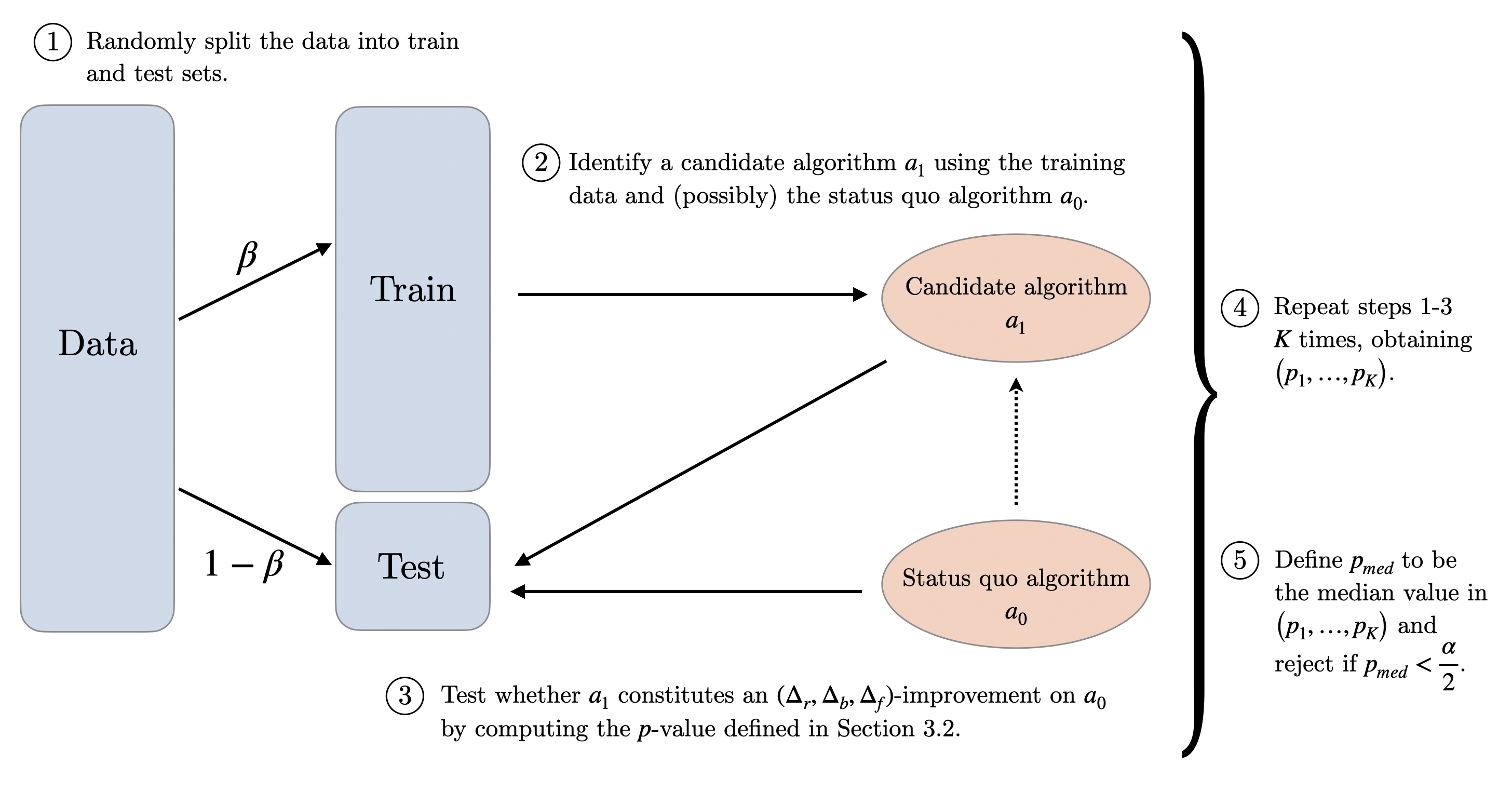}
    \end{center}
    \caption{A summary of our procedure for testing $(\Delta_r,\Delta_b,\Delta_f)$-improvability.} \label{fig:Flowchart}
\end{figure}

\smallskip

This proposed procedure is summarized in Figure \ref{fig:Flowchart}. When the null hypothesis is rejected, we conclude that the status quo algorithm $a_0$ is $(\Delta_r,\Delta_b,\Delta_f)$-improvable within the specified class of algorithms. This procedure does not identify a specific algorithm that achieves that improvement, but instead implies that the selection rule $\rho$ can find one. If the analyst has shown $(\Delta_r,\Delta_b,\Delta_f)$-improvability and requires a single algorithm to use as a substitute for the status quo algorithm, we recommend applying the selection rule $\rho$ to the entire dataset and using its output.

We now turn to supporting details and results for step 3 of our procedure.

\subsection{Testing whether $\hat{a}^{\rho}_{1n}$ constitutes a $(\Delta_r,\Delta_{b},\Delta_f)$-improvement on $a_0$.} \label{sec:Test} 

In this section we propose a hypothesis testing procedure to evaluate whether a fixed algorithm $a_1\equiv \hat{a}^{\rho}_{1n}$ constitutes an improvement relative to the status quo algorithm $a_0$.  Since $a_1$ was selected using the training sample $S_{train}$, it suffices to describe a test of $(\Delta_r, \Delta_b, \Delta_f)$-improvement for a fixed algorithm $a_1$ relative to $a_0$, using the test sample $S_{test}^n = \{(X_i,Y_i,G_i)\}_{i=1}^{\ell_n}$. Throughout we use $\ell_n = n - m_n$ to denote the size of the test sample. To that end, we test the null hypothesis
\begin{equation}\label{eq:H0_split}H_0: A_{1r} \leq A_{0r}(1 + \Delta_r) \text{ OR } A_{1b} \leq A_{0b}(1 + \Delta_{b}) \text{ OR } |F_{1r} - F_{1b}| \ge |F_{0r} - F_{0b}| (1 - \Delta_f)~,
\end{equation}
against the alternative
\[H_1:  A_{1r} > A_{0r}(1 + \Delta_r) \text{ AND } A_{1b} > A_{0b}(1 + \Delta_{b}) \text{ AND } |F_{1r} - F_{1b}| < |F_{0r} - F_{0b}| (1 - \Delta_f)~,\]
for an arbitrary vector $\Delta \equiv (\Delta_r,\Delta_{b},\Delta_{f}) \in \mathbb{R}^3$, where \[(A_{0r},A_{0b},F_{0r},F_{0b}) \equiv (U_A^{r}(a_0),U_A^{b}(a_0), U_F^r(a_0), U_F^{b}(a_0))\] denotes the accuracy and fairness levels under the status quo algorithm $a_0$ and \[(A_{1r},A_{1b},F_{1r}, F_{1b}) \equiv (U_A^r(a_1),U_A^{b}(a_1),U_F^r(a_1), U_F^{b}(a_1))\] denotes the accuracy and fairness levels under the candidate algorithm $a_1$.

Since $H_0$ is formulated as a union of three conditions, a test of $H_0$ can be constructed by first constructing tests for each of the conditions individually, and then combining the individual tests using the intersection-union method.\footnote{See \citet{{casella2021statistical}} for a textbook reference on the intersection-union method. Our Proposition \ref{prop:IUT} in the appendix extends Theorems 8.3.23-24 in \citet{{casella2021statistical}} to an asymptotic setting. }
Specifically, we will construct tests for each of the individual conditions that make up $H_0$, and reject $H_0$ if all three tests reject at once.\footnote{We consider a test that is the union of three individual tests because it leads to a transparent test that is easy to interpret and analyze. In contrast, combining the three conditions that make up $H_0$ into one test directly would require us to make ex-ante arbitrary decisions about how to weight the different conditions, where it would be relatively difficult to understand how these weighting decisions affected the statistical properties of the test.}

First, we specify test statistics for each of the individual tests that make up $H_0$. To simplify the derivations, we introduce a slight change in notation. Let $Z_i = (X_i, Y_i, G_i)$ and define
\begin{align*}
    u_{tg}^A(Z_i, \theta) & = \frac{1}{\theta^{A}_{tg}}\left(u_A(X_i, Y_i,a_t(X_i)) \cdot \mathbbm{1}(G_i=g)\right) \\ u_{tg}^F (Z_i, \theta) & = \frac{1}{\theta^F_{tg}}\left(u_F(X_i, Y_i, a_t(X_i)) \cdot \mathbbm{1}(G_i=g)\right)
\end{align*}
where 
\begin{align*}
    \theta^A_{tg} & = E\left[w_A(X_i, Y_i,a_t(X_i)) \cdot \mathbbm{1}(G_i=g)\right] \\
    \theta_{tg}^F & = E\left[w_F(X_i, Y_i, a_t(X_i)) \cdot \mathbbm{1}(G_i=g)\right]
\end{align*}
The vector $\theta\in \mathbb{R}^J$ collects all of the unknown nuisance parameters $(\theta^A_{tg}, \theta^F_{tg})$ which appear as normalizations in the utility specifications. 
This re-writing allows for the accuracy and fairness criteria $A_{tg}$ and $F_{tg}$ to be expressed as unconditional expectations of the functions $u_{tg}$, rather than as ratios of the conditional expectations of $u_A$, $u_F$, $w_A$, $w_F$ (as defined in (\ref{eq:UAUF})). That is, $A_{tg} =  E[u^A_{tg}(Z_i,\theta)]$ and $F_{tg} =  E[u^F_{tg}(Z_i,\theta)]$. 

In the results that follow, we will in fact allow $u^A_{tg}$ and $u^F_{tg}$ to depend on the nuisance parameter $\theta$ more generally than what is considered above (see specifically Assumption \ref{assp:Utilities}). Since $\theta$ is unknown, it will need to be estimated. We assume that $\theta$ can be written as $\theta = E[h_{i}]$ where $h_i =  h(Y_i, G_i, X_i, a_{1}(X_i),a_{0}(X_i))$ for some known function $h: \mathcal{Y}\times G \times \mathcal{X} \times \mathcal{D}^{2} \to \mathbb{R}^{J}$.\footnote{This imposes few restrictions, and  encompasses our four previous examples in Section \ref{sec:Framework}.} For instance, consider the following example based on the calibration utility in Example \ref{ex:Calibration}.

\begin{example}\label{ex:theta}
Suppose we wish to test the null hypothesis $H_0$ in a setting where $u_A = u_F = u$, and $u$ is  the calibration utility function defined in Example \ref{ex:Calibration}. Then $u^A_{tg} = u^F_{tg} = u_{tg}$ is given by
\[u_{tg}(Z_i,\theta) = \frac{1}{\theta_{tg}}Y_i\mathbbm{1}\{a_t(X_i) = 1\}\mathbbm{1}\{G_i = g\}  \]

where $\theta= (\theta_{0r},\theta_{1r},\theta_{0b},\theta_{1b}) = E[h(G_i,a_1(X_i),a_0(X_i)]$ and
\[h(G_i,a_1(X_i),a_0(X_i))=\left(\begin{array}{c}
\mathbbm{1}(G_i=r,a_0(X_i)=1) \\
\mathbbm{1}(G_i=r,a_1(X_i)=1) \\
\mathbbm{1}(G_i=b,a_0(X_i)=1) \\
\mathbbm{1}(G_i=b,a_1(X_i)=1)\end{array}\right)~.\]
\end{example}

We estimate $A_{tg}$, $F_{tg}$, and $\theta$ using their empirical analogs $\hat{A}_{tg} = \frac{1}{\ell_n}\sum_{i} u^A_{tg}(Z_i,\hat{\theta})$, $\hat{F}_{tg} = \frac{1}{\ell_n}\sum_{i}u^F_{tg}(Z_i,\hat{\theta})$ and $\hat{\theta} = \frac{1}{\ell_n}\sum_{i} h_i$. Our test statistics for the individual tests are then defined as follows. When testing whether $a_1$ is less accurate than $a_0$ for groups $r$ and $b$, let
\[\hat{T}_{r,n} \equiv (\hat{A}_{1r} - (1 + \Delta_r)\hat{A}_{0r})~,\] 
\[\hat{T}_{b,n} \equiv (\hat{A}_{1b} - (1 + \Delta_b)\hat{A}_{0b}) ~.\] 
When testing whether $a_1$ is less fair than $a_0$, define
\[\hat{T}_{f,n} \equiv \left(\big|\hat{F}_{1r} - \hat{F}_{1b}\big| - (1 - \Delta_f)\big|\hat{F}_{0r} - \hat{F}_{0b}\big|\right)\]

 We propose using the nonparametric bootstrap to generate the critical values of our test, since this eliminates the need to compute  standard errors on a case-by-case basis for each utility function. Specifically, if we fix $\alpha$ at some nominal level, then for the accuracy parts of the test, the rejection rule is given by $\phi_n^{(g)} = \mathbbm{1}\{\sqrt{\ell_n}\hat{T}_{g,n} > c^*_{g,1 - \alpha}\}$ for $g \in \{r,b\}$, where 
$c^*_{g,1 - \alpha} = \Psi_{g,n}^{-1}(1-\alpha)$ is the $(1 - \alpha)$-quantile of the bootstrap distribution $\Psi_{g,n}$ defined in Appendix \ref{sec:boot}. For the fairness part of the test, the rejection rule is given by $\phi_n^{(f)} = \mathbbm{1}\{\sqrt{\ell_n}\hat{T}_{f,n} < c^*_{f,\alpha}\}$ where $c^*_{f,\alpha} = \Psi_{f,n}^{-1}(\alpha)$ is the $\alpha$-quantile of the bootstrap distribution $\Psi_{f,n}$ defined in Appendix \ref{sec:boot}. Correspondingly, the $p$-value associated with the test $\phi_n^{(g)}$ is given by $1 - \Psi_{g,n}(\hat{T}_{g,n})$ and the $p$-value associated with the test $\phi_n^{(f)}$ is given by $\Psi_{f,n}(\hat{T}_{f,n})$. Our test of the null hypothesis \eqref{eq:H0_split} is finally given by \[\phi_n(a_0, a_1) = \phi_n^{(r)}\cdot\phi_n^{(b)}\cdot\phi_n^{(f)},\] so  we reject if and only if all three component tests reject. As a result, the $p$-value associated with the test $\phi_n$ is given by the maximum of the three $p$-values associated with each of the component tests.

The validity of our bootstrap procedure hinges crucially on the structure of our specific problem. Indeed, given the form of $\hat{T}_{f,n}$, it is well known that a standard nonparametric bootstrap fails to be consistent at points of non-differentiability \citep[see for instance][]{fang2019inference}. This does not affect the validity of our test for the following reason: Recall that the null is a union of three component null hypotheses. If the null holds because $A_{1g} \le A_{0g}(1+\Delta_g)$ for some $g \in \{r,b\}$, then (as we show in the proof of Theorem \ref{thm:full_test}) it is not necessary for the fairness test $\phi_n^{(f)}$ to control size, since one of the other two tests will. Thus, the only region of the null space where the size of $\phi_n^{(f)}$ is relevant is in the region where \[\vert F_{1r} - F_{1b}\vert \geq  \vert F_{0r} - F_{0b}\vert (1+\Delta_f)~.\] But Assumption \ref{assp:n} implies that the RHS is strictly positive. Thus in this region of the null space,  $\vert F_{1r} - F_{1b}\vert$ is strictly positive, and we are consequently bounded away from points of non-differentiability of $\hat{T}_{f,n}$.\footnote{Under the alternative, we prove consistency of the bootstrap under a mirrored assumption that the selection rule selects candidate algorithms which are not arbitrarily fair; we return to this point in the discussion following Theorem \ref{thm:consistent_test}.}

\section{Main Results}\label{sec:Full_test}
We now show that under certain regularity conditions the test proposed in Section \ref{sec:Approach} is asymptotically valid when the candidate algorithm $a_1$ is selected using the data in $S^{n}_{train}$. We will also show that the test is consistent under an additional assumption on the selection rule $\rho(\cdot)$ which we call \emph{improvement convergence}.

First consider $K=1$, i.e., there is only one round of sample-splitting. Let $(S_{train}^n,S_{test}^n)$ denote the realized split of training and testing samples of the data $\{(X_i, Y_i, G_i)\}_{i=1}^n$. Denote the candidate algorithm selected using rule $\rho(\cdot)$ on the training sample $S_{train}^n$ as
\[\hat{a}^\rho_{1n} \equiv \rho(S_{train}^n).\]
We show that the test $\phi_n(a_0, \hat{a}^\rho_{1n})$ (as defined in Section \ref{sec:Test}) is asymptotically valid. Our result holds under the following regularity conditions restricting the distribution $P$ and class of algorithms $\mathcal{A}$.

\begin{assumption} \label{assp:Utilities} The nuisance parameters $\theta$ belong to an open convex subset $\Theta \subset \mathbb{R}^{J}$. The functions $u_{tg}^{A}(z,\cdot)$ and $u_{tg}^{F}(z,\cdot)$ are twice continuously differentiable on $\Theta$ for every $z \in \mathcal{Z}\equiv\mathcal{X} \times \mathcal{Y} \times \mathcal{G} $, $t \in \{0,1\}$, and $g \in \{r,b\}$. Furthermore, $h(\cdot)$ is bounded, and  $u_{tg}^{A}$, $u_{tg}^{F}$, and their first and second-order partial derivatives (with respect to $\theta$) are all uniformly bounded over $(z, \theta) \in \mathcal{Z}\times \Theta$. 
\end{assumption}

\begin{assumption} \label{assp:NonDegenerate} 
Let $\mathcal{M}$ denote the set of covariance matrices of the nuisance parameters and utilities (as defined in \eqref{eq:sigma} in the appendix) for every candidate algorithm $a_1 \in \mathcal{A}$. Then we assume
\[\inf_{\Sigma \in \mathcal{M}}\lambda_{\rm min}(\Sigma) \ge \nu~,\]
for some $\nu > 0$, where $\lambda_{\rm min}(\Sigma)$ denotes the smallest eigenvalue of $\Sigma$.
\end{assumption}

Assumption \ref{assp:Utilities} allows us to appropriately linearize the utility functions as a function of $\theta$, uniformly over the parameter space $\Theta$. The uniformity of the linearization is important since the true value of the nuisance parameter can fluctuate as a function of sample size. The boundedness conditions in Assumption \ref{assp:Utilities} impose implicit constraints on the parameter space $\Theta$ and the support of the data. For instance, revisiting Example \ref{ex:theta}, we see that the support of $Y_i$ should be bounded, and that each component of $\theta$ should be bounded away from zero. Assumption \ref{assp:NonDegenerate} ensures that the limiting variances of our test statistics are non-degenerate uniformly in the space of candidate algorithms $\mathcal{A}$. Once again, this uniform non-degeneracy is important because the candidate algorithm is allowed to vary arbitrarily as a function of sample size. Low-level conditions which guarantee Assumption \ref{assp:NonDegenerate} are difficult to articulate at this level of generality, but Assumption \ref{assp:NonDegenerate} intuitively requires that the variances of the utilities be bounded away from zero uniformly in the space of candidate algorithms, and that the correlation in the utilities between the status-quo and candidate algorithms be bounded away from one. We expect these conditions to be satisfied for most selection rules used in practice.\footnote{Note that, more precisely, we in fact only require Assumption \ref{assp:NonDegenerate} to hold for the set of algorithms $\mathcal{A}' \subseteq \mathcal{A}$ which can be realized by the selection rule $\rho(\cdot)$.}

Theorem \ref{thm:full_test} establishes that the test $\phi_n$ is asymptotically of level $\alpha$ regardless of the choice of algorithms $\hat{a}^\rho_{1n}$, as long as the nuisance parameters and utilities are well-behaved over the class $\mathcal{A}$ as described in Assumptions \ref{assp:n}, \ref{assp:Utilities} and \ref{assp:NonDegenerate}.
\begin{theorem}\label{thm:full_test}
Suppose $P$ and $\mathcal{A}$ satisfy Assumptions \ref{assp:n}, \ref{assp:Utilities} and \ref{assp:NonDegenerate}, and suppose the null hypothesis given in \eqref{eq:H0_full} holds. Then \[\limsup_{n \rightarrow \infty}E_{P}[\phi_n(a_0,\hat{a}^\rho_{1n})] \le \alpha~.\]
\end{theorem}

We recommend combining the test results over multiple splits of the data to reduce the unpredictability in the procedure induced by sample splitting (see Section \ref{sec:RobustManipulation} for a theoretical justification).  Formally, consider $K>1$ and let $(S_{train}^{n,k},S_{test}^{n,k})$ denote the realized split of training and testing samples in the $k$-th round of sample splitting. Further let 
\[\hat{a}^\rho_{1nk} \equiv \rho(S_{train}^{n,k})\]
denote the candidate algorithm selected using rule $\rho$ on the training sample $S_{train}^{n,k}$.
Define the test statistic $\psi_n = \mathbbm{1}\left\{\frac{1}{K}\sum_{k = 1}^K \phi_{n,k}(a_0,\hat{a}^\rho_{1nk}) \ge 1/2\right\}$ to reject if the median test statistic (across the rounds of sample splitting) rejects.
We establish the following corollary of Theorem \ref{thm:full_test}:

\begin{corollary}\label{cor:split_test}
Under the assumptions of Theorem \ref{thm:full_test},
\[\limsup_{n \rightarrow \infty}E_P[\psi_n] \le 2\alpha~.\]
\end{corollary}

Note that, as a consequence of Corollary \ref{cor:split_test}, to ensure that the test $\psi_n$ is asymptotically of level $\alpha$ we require that the component tests $\phi_{n,k}$ for $k = 1, \ldots, K$ are of level $\alpha/2$. Equivalently, we reject the null if the median $p$-value (across the rounds of sample splitting) satisfies $p_{\mathrm{med}} < \alpha/2$. We would in general expect this conservativeness to result in a loss of power relative to a test performed using a single split, but the local power analysis conducted in  \cite{diciccio2020exact} suggests that this may not always be the case if the alternative is not too close to the null. In particular, their Theorem 3.1 demonstrates this for a one-sided test of the mean, which is similar to the component nulls in \eqref{eq:H0_split}.

Finally, we establish the consistency of the test $\phi_{n}(a_0, \hat{a}^\rho_{1n})$ under high-level conditions on the behavior of the selection rule $\rho$. 

\begin{definition}\label{assp:Limit} Suppose the null hypothesis given in \eqref{eq:H0_full} does not hold. We say that the selection rule $\rho$ is \emph{improvement-convergent} if, for each $g \in \{r, b\}$, $U^{g}_A(\hat{a}^\rho_{1n}) \xrightarrow{P} U^{g}_A(\gamma)$ and $U^{g}_F(\hat{a}^\rho_{1n}) \xrightarrow{P} U^{g}_F(\gamma)$ for some algorithm $\gamma \in \mathcal{A}$ that constitutes a $(\Delta_r, \Delta_b, \Delta_f)$-improvement over the algorithm $a_0$. 
\end{definition}

In words, improvement convergence says that whenever the algorithm $a_0$ is $(\Delta_r,\Delta_b,\Delta_f)$-improvable in $\mathcal{A}$, the selection rule $\rho$ is guaranteed to find such an improvement given enough data. Theorem \ref{thm:consistent_test} establishes that the test $\phi_n$ is consistent (that is, has power approaching one for any distribution in the set of alternatives) when the selection rule $\rho$ is improvement-convergent.

\begin{theorem}\label{thm:consistent_test}
Suppose $P$ and $\mathcal{A}$ satisfy Assumptions \ref{assp:n}, \ref{assp:Utilities} and \ref{assp:NonDegenerate}, and let $P$ be any distribution such that the null hypothesis given in $\eqref{eq:H0_full}$ does not hold. For any improvement-convergent selection rule $\rho$ for which the selected candidate algorithm satisfies $|U^r_F(\hat{a}^{\rho}_{1n}) - U^b_F(\hat{a}^{\rho}_{1n})| \ge \xi$ almost surely for some fixed $\xi > 0$, we have
 \[\lim_{n \rightarrow \infty}E_P\left[\phi_n\left(a_0, \hat{a}^\rho_{1n}\right)\right] = 1.~\]
\end{theorem}

This result assumes that the selection rule $\rho$ does not select a candidate algorithm that is arbitrarily fair, even asymptotically. This is related to our discussion at the end of Section \ref{sec:Test}: if the algorithms chosen by the selection rule converge to perfect fairness, then our bootstrap procedure is no longer consistent, and our proof strategy does not establish consistency of our test in this case. However, we emphasize that bootstrap consistency is \emph{not} a necessary condition for our test to have reasonable power, and we demonstrate via simulation in Appendix \ref{sec:sims} that our test displays reasonable power even when employing candidate algorithms that are perfectly fair.

We expect that, for many datasets, it will be possible to reject the null using a naive selection rule that does not satisfy improvement convergence, as is the case in our empirical illustration in Section \ref{sec:Application}. However, Theorem \ref{thm:consistent_test} guarantees that the null will be asymptotically rejected (in cases where it should be) when the selection rule satisfies this property. Many of the methods developed in the literature on identifying less discriminatory algorithms (as reviewed in Section \ref{subsec:RelatedLit}) are improvement-convergent for specific utility specifications and algorithm classes. In Appendix \ref{sec:Search} we propose selection rules constructed using mixed integer linear programs, which we conjecture are improvement-convergent when $\mathcal{A}$ is the set of linear classifiers.

\section{Microfoundation for Repeated Sample-Splitting} \label{sec:RobustManipulation}

Our recommended approach is to compute the median $p$-value across $K>1$ repeated sample splits and to reject when this $p$-value is less than $\alpha/2$ (where $\alpha$ is the level of the test). One could alternatively conduct a single sample-split and reject if the resulting $p$-value is less than $\alpha$. As our 
Theorem \ref{thm:full_test} and Corollary \ref{cor:split_test} demonstrate, both procedures are valid, and we generally expect that the use of a single split should lead to a less conservative test. However, because the resulting $p$-value can vary substantially across different splits, relying on a single train-test split introduces arbitrariness---what \citet{meinshausen2009p} term the `$p$-value lottery.' This sensitivity to the sample split introduces the possibility of  manipulation. As \citet{ritzwoller2023reproducible} note: ``Researchers are incentivized to report significant results. If there is scope to materially alter the statistics that they report through the choice of the split of their sample, should this choice be left to chance?'' This concern is particularly relevant to disparate impact claims, where complainants may have incentives to find fairness violations even when they do not exist.

In this section, we thus formalize the intuition that repeated sample-splitting provides stronger safeguards against manipulation than single splits, thereby providing a theoretical justification for this methodological choice. Our microfoundation is a game between a policymaker who sets the statistical procedure and an analyst who can (secretly) repeat the procedure multiple times. Within this framework, we define what it means for a test to be robust to manipulation. Section \ref{sec:RMmodel} presents this model, and Section \ref{sec:RMresults} proves that repeated sample-splitting is more robust to manipulation than single train-test splits. While our motivation comes from disparate impact testing, these results apply to any valid test, making them relevant beyond our specific setting.

\subsection{Model} \label{sec:RMmodel}
Suppose a firm's algorithm is the subject of a potential disparate impact case. We consider a game between two players: an \emph{analyst} auditing the firm and a \emph{policymaker}.  There is a fixed statistical test of exact size $\alpha \in (0,1)$, which produces a $p$-value from any given train and test split. The policymaker chooses between two procedures $s_1$ and $s_2$ based on this test. The first corresponds to a single train-test split, where the null of no fairness-improvability is rejected if the resulting $p$-value is less than $\alpha$. The second corresponds to $K$ train-test splits, where the null is rejected if the median $p$-value across these splits is less than $\alpha/2$ (i.e., our proposed method).  The analyst observes which  procedure is chosen, and repeats that procedure $m \in \mathbb{Z}_+$ times at a cost of $c_1(m)$ for procedure 1 or $c_2(m)$ for procedure 2, where $c_1$ and $c_2$ are increasing and weakly convex functions. (For example, the specifications $c_1(m) = \gamma m$ and $c_2(m) = \gamma K m$ would give the analyst a constant cost $\gamma$ for each train-test split.) The analyst reports the $p$-value from one of these repetitions, and this reported $p$-value determines whether the null is rejected. Crucially, we suppose that the hypothesis test is interpreted as if $m=1$, which is the standard convention.\footnote{This distinguishes our model from related works such as \citet{Henry2009}, \citet{FelgenhauerLoerke}, \citet{HenryOttaviani}, and \citet{Herresthal2022}, in which an agent gathers information through hidden testing or experimentation, and subsequently chooses whether to disclose his findings. In those papers, the principal or sender (the equivalent of our ``policymaker'') updates his beliefs about an unknown payoff-relevant state given the agent's report and equilibrium strategy. In our model, payoffs are instead directly determined by whether the null is rejected at the reported $p$-value.}

    We are interested in settings where the analyst would like to conclude that the algorithm is fairness-improvable even when it is not. With this kind of manipulation in mind, we condition on the state of the world in which the null hypothesis holds. The analyst's payoff is $1-c_i(m)$ if the null is rejected under procedure $s_i$ and $-c_i(m)$ otherwise, while the policymaker's payoff is 0 if the null is rejected and 1 otherwise.

To summarize the timeline:
\begin{enumerate}
\item The policymaker chooses between statistical procedures $s_1$ and $s_2$.
\item The analyst chooses $m \in \mathbb{Z}_+$.
\item The $p$-values from the $m$ repetitions of the procedure are realized, and the analyst chooses one to report.
\item Payoffs are realized.
\end{enumerate}
We consider subgame-perfect equilibria.

To simplify notation, it will be useful to couple the realizations of the $p$-values under the two tests in the following way: For each repetition $i$ of the second statistical test, let $(p_i^1, \dots, p_i^K)$ denote the random $p$-values across the $K$ iterations of train-test. For each repetition $i$ of the first statistical test, let $p_i^1$ be the $p$-value in the single train-test split. Further define $\phi_{ik}^\alpha=\mathbbm{1}(p_i^k < \alpha)$ to be an indicator for whether the $p$-value from the $(i,k)$-th train-test split leads to rejection at an $\alpha$ significance level. By assumption that the statistical test for each train-test split has exact size, each random variable $\phi_{ik}^\alpha$ is Bernoulli with parameter $\alpha$. 

Equilibrium actions and payoffs can be determined by backwards induction. In period 2, the analyst solves
\begin{equation} \label{eq:analystProblem1}
\max_{m \in \mathbb{Z}_+} \left[P\left(\max_{1\leq i \leq m } \phi_{i1}^\alpha = 1\right) - c_1(m)\right]
\end{equation}
if procedure $s_1$ was chosen, or
\[\max_{m \in \mathbb{Z}_+} \left[P\left(\max_{1\leq i \leq m } \mathbbm{1}\left(\frac1K \sum_{k=1}^K \phi_{ik}^{\alpha/2}\geq \frac12 \right)=1\right) - c_2(m)\right]\]
if procedure $s_2$ was chosen. Let $m_1^*$ and $m_2^*$ respectively denote the solutions to these two problems (breaking ties in favor of the smallest number of repetitions). Then in period 1, the policymaker's expected payoffs from the two procedures are
\[u_P(s) = \left\{ \begin{array}{cc}
1 - \displaystyle P\left(\max_{1\leq i \leq m_1^* } \phi_{i1}^\alpha =1\right) & \mbox{if $s=s_1$}\\
 1 - \displaystyle P\left(\max_{1\leq i \leq m_2^* } \mathbbm{1}\left(\frac1K \sum_{k=1}^K \phi_{ik}^{\alpha/2}\geq \frac12 \right)=1\right) & \mbox{if $s=s_2$}
 \end{array}\right.\]

We define robustness to manipulation as follows.

\begin{definition}
Say that $s_i$ is \emph{more robust to manipulation} than $s_j$ if $u_P(s_i) > u_P(s_j)$.
\end{definition}

\noindent That is, $s_i$ is more robust to manipulation than $s_j$ if the policymaker's expected payoff from choosing $s_i$ is strictly higher. Equivalently, the probability that the null is incorrectly rejected---given the analyst's endogenous choice of how many times to re-run the procedure---is higher under $s_i$ than $s_j$. Finally, note that when $s_i$ is more robust to manipulation than $s_j$, then the policymaker chooses $s_i$ in every equilibrium.

\subsection{Results} \label{sec:RMresults}

Our main result in this section is that when $K$ is sufficiently large, then procedure $s_2$ (repeated sample-splitting) is more robust to manipulation than $s_1$ (a single train-test split) under the following assumption.

\begin{assumption} \label{assp:LowCost}
$P(\phi_{21}^\alpha = 1 \mid \phi_{11}^\alpha = 0) >\frac{c_1(2)-c_1(1)}{1-\alpha}$.
\end{assumption}

This assumption lower bounds the probability that the second repetition of $s_1$ yields a different outcome from the first. It implies that the marginal benefit of the second repetition is higher than its marginal cost, and thus ensures that the analyst will repeat the first procedure at least once. This assumption is satisfied so long as tests are not perfectly correlated, and the cost of the second repetition is not too large. (For example, if $c_1(m) = \gamma m$ is linear, then the bound can be restated as $\gamma < (1-\alpha) P(\phi_{21}^\alpha = 1 \mid \phi_{11}^\alpha = 0)$.)

\begin{theorem} \label{prop:RobustManipulationGeneral}  Suppose Assumption \ref{assp:LowCost} holds. Then for $K$ sufficiently large, procedure $s_2$ is more robust to manipulation than $s_1$.
\end{theorem}

The proof of this result proceeds in two parts. First, we show that
\begin{equation} \label{eq:Part1}
P\left(\max_{1\leq i \leq m } \phi_{i1}^\alpha = 1\right)  \geq \alpha
\end{equation}
for every $m$, and moreover the inequality holds strictly for every $m>1$ (where Assumption \ref{assp:LowCost} guarantees that the analyst's solution satisfies $m_1^*>1$). This is extremely intuitive, and says that the probability of rejecting the null given $m>1$ repetitions of the first procedure is larger than with just one. (The only way it could be otherwise is if the test was perfectly correlated across repetitions, which Assumption \ref{assp:LowCost} rules out.)

The other half of the proof demonstrates that for $K$ sufficiently large
\begin{equation} \label{eq:Part2}
P\left(\max_{1\leq i \leq m } \mathbbm{1}\left(\frac1K \sum_{k=1}^K \phi_{ik}^{\alpha/2}\geq \frac12 \right)=1\right) \leq \alpha.
\end{equation}
for every $m$ not too large (including $m=m_2^*$). Intuitively, for every fixed repetition $i$, the sample average $\frac1K \sum_{k=1}^K \phi^{\alpha/2}_{ik}$ approaches $P(\phi_{ik}^{\alpha/2}=1)$ as $K$ grows large. So when the probability of rejection in any single $(i,k)$ train-test split is small, it is unlikely that the test is rejected in more than half of the $K$ iterations, and thus unlikely that the median test statistic rejects.  This intuition is straightforwardly formalized when the tests $\phi^{\alpha/2}_{ik}$ are i.i.d, but more care is required in our setting where $\phi^{\alpha/2}_{ik}$ are instead positively correlated. The key observation is that the tests $\phi^{\alpha/2}_{ik}$ are exchangeable, and thus (by de Finetti's theorem) the tests $\phi^{\alpha/2}_{ik}$ are i.i.d Bernoulli($q$) conditional on the realization of the random Bernoulli parameter $q$. What this means is that for each realization of $q$, we can apply the strong law of large numbers to conclude that the sample average $\frac1K \sum_{k=1}^K \phi^{\alpha/2}_{ik}$ approaches $q$. As $K$ grows large, the ex-ante probability that half of these tests reject is then simply the probability that $q\geq 1/2$, which we show is no more than $\alpha$. 

These arguments formalize the intuitive sense in which it is easier to manipulate a test that depends only on one repetition of a procedure (in which case the analyst only needs to get lucky once), than a test which aggregates outcomes across many repetitions of that procedure. Our result does not use any properties of the test $\phi^{\alpha/2}_{ik}$ beyond having exact size and that the set of tests $\{\phi^{\alpha/2}_{ik}: 1 \le k \le K\}$ are exchangeable, and thus extends more broadly to many testing procedures which employ sample splitting.  

Assumption \ref{assp:LowCost} is made for convenience, but is not critical for our qualitative conclusion. Specifically, if instead the policymaker optimally runs the first procedure only once (i.e., $m_1^*=1$), then our assumption that the test has exact size implies $P\left(\phi_{11}^\alpha = 1\right) = \alpha$. Moreover, our proof of Theorem \ref{prop:RobustManipulationGeneral} shows that  for every $\varepsilon$, there is a $\overline{K}$ sufficiently large such that 
\[P\left(\max_{1\leq i \leq m_2^* } \mathbbm{1}\left(\frac1K \sum_{k=1}^K \phi_{ik}^{\alpha/2}\geq \frac12 \right)=1\right)  \leq \alpha  + \varepsilon\]
for every $K > \overline{K}$. Thus when Assumption \ref{assp:LowCost} fails, then for large $K$ either the policymaker prefers the second procedure (reproducing our original result) or the policymaker is approximately indifferent between the two procedures. In fact, we expect that our original result can be extended to cover the case when Assumption \ref{assp:LowCost} fails, but do not pursue it here.

 Finally, although an explicit bound for $K$ is not possible without stronger assumptions, it is instructive to consider the edge case of i.i.d $p$-values. This i.i.d benchmark is independently interesting because the analyst's incentives for repeating the test are higher than in our setting of positively correlated $p$-values.\footnote{For example, consider the probability of rejecting the null under $m$ repetitions of the first procedure. By de Finetti's theorem, $P(\max_{1\leq i \leq m} \phi_{i1}^\alpha = 1)=1-\int_{q \in [0,1]} (1-q)^m d\mu(q)$ for some measure $\mu$, where $\mathbb{E}_\mu(q)=\alpha$. And by Jensen's inequality, $1-\int_{q \in [0,1]} (1-q)^m d\mu(q) \leq 1 - (1-\mathbb{E}_\mu(q))^m = 1-(1-\alpha)^m$, where $1-(1-\alpha)^m$ is the probability of rejecting when the $p$-values are instead i.i.d. Thus for every $m$, the probability of (incorrectly) rejecting the null is higher given i.i.d $p$-values than positively correlated ones. See the proofs of Theorem \ref{prop:RobustManipulationGeneral} and Proposition \ref{prop:RobustManipulation} for more detail.} Within this ``maximal incentive for manipulation'' benchmark, we can strengthen our previous result to the following.

 \begin{proposition} \label{prop:RobustManipulation} Suppose the infinite array of $p$-values $(p_{i}^k)$ has i.i.d entries. Then statistical procedure $s_2$ is more robust to manipulation than $s_1$ for every $K > -\frac{2\ln \alpha}{(1-\alpha)^2}$.
 \end{proposition}

 This bound on $K$ is extremely undemanding. For example, when $\alpha=0.05$, the result holds for any $K \geq 7$. 
 
\section{Empirical Application} \label{sec:Application}

We conclude by illustrating our approach in the setting of \citet{Obermeyer2019-te}. The status quo algorithm is a commercial healthcare algorithm used by a large academic hospital to identify patients to target for a high-risk care management program. The algorithm assigns to each patient a risk score, and identifies those patients with risk scores above the 97th percentile for automatic enrollment in the program. We use our proposed approach to evaluate the improvability of the hospital's algorithm, finding that it is strictly FA-dominated within the class of linear classifiers. We further quantify the size of possible improvements by testing the null hypotheses that the hospital's algorithm is not $(\delta_a,\delta_a,\delta_f)$-improvable for different values of $\delta_a$ and $\delta_f$. We find that large improvements in fairness are possible without compromising on accuracy, while we are unable to establish that the reverse is true.

\subsection{Data and Classification Problem}  The data includes observations from 48,784 patients, among which 43,202 self-report as White and 5,582 self-report as Black. Following \citet{Obermeyer2019-te}, we take these to be the two group identities, denoted $g \in \{w,b\}$. Each patient $i$'s covariate vector $X_i \in \mathcal{X}$ includes 8 demographic variables (e.g., age and gender), 34 comorbidity variables (indicators of specific chronic illnesses), 13 cost variables (claimed cost broken down by type of cost), and 94 biomarker and medication variables.  Although group identity $g$ is known for each patient in the dataset, the status quo algorithm omits this variable for prediction, so we do as well. Finally, the data set reports each patient $i$'s total number of active chronic illnesses in the subsequent year, which is interpreted as a measure of the patient's true health needs.\footnote{
In principle, enrollment into the high-risk care management program may affect the outcome variable $Y$. Although we expect that a single year of enrollment in this program will have only minor effects on the patients' number of chronic illnesses (which $Y$ measures), we conduct a robustness check in Appendix~\ref{app:program_effect} in which we repeat our main analysis using the total number of active chronic illnesses in the same year as the enrollment decision. All results are qualitatively similar.
}
We take this to be the patient's type $Y_i \in \mathcal{Y}$.

An algorithm assigns each patient a decision $D_i \in \{0,1\}$ based on the covariate vector $X_i$, where 1 corresponds to automatic enrollment into the care management program. We respect the capacity constraint of the status quo algorithm by restricting attention to algorithms that select 3\% of the population for automatic enrollment in expectation.

Following \citet{Obermeyer2019-te}, we evaluate both accuracy and fairness using the calibration utility function introduced in Example \ref{ex:Calibration}, i.e.
\begin{equation}\label{eq:AccuracyCalibration}
U^g(a) = E[Y \mid a(X)=1, G=g]
\end{equation}
and
\[\vert U^b(a) - U^w(a) \vert =  \vert E[Y \mid a(X)=1, G=b] - E[Y \mid a(X)=1, G=w] \vert.\]

That is, an algorithm is more accurate if the expected number of health conditions is higher among both Black and White patients assigned to the program. It is more fair if it reduces the disparity in the expected number of health conditions among Black and White patients assigned to the program.\footnote{Here and throughout, we drop the notational distinction between $U_A^g$ and $U_F^g$, as they are identical in this example.}

\subsection{Results}
\label{sec:empirical_result}
We implement our proposed approach in Section \ref{sec:Approach} to test the fairness- and accuracy- improvability of the hospital's algorithm. We use $K=7$ iterations and select a $\beta=1/2$ fraction of the data to use for training. The hospital's algorithm provides each patient with a risk score. To enroll patients to treatment according to the status quo algorithm, we first designate the $97$th percentile of the risk score distribution in the training data as our admittance threshold. Then, in the testing data, we automatically enroll ($D=1$) those patients whose risk score is above this threshold. Our candidate algorithms similarly assign risk scores to patients, and enroll those patients whose risk scores are above the $97$th percentile threshold computed in the training data, but use different risk score calculations compared to the status quo algorithm.

Specifically, to identify candidate algorithms, we first use the training data to compute an alternative risk score assignment rule
$\hat{f}:\mathcal{X} \rightarrow \mathcal{Y}$, which maps each covariate vector into a predicted number of health conditions. We then set the 97th percentile of the computed risk scores $\hat{f}$ from the training data as our admittance threshold. Next, we compute the risk scores for all of the patients in the test data, and select those patients whose risk score is above the threshold.

We consider several different methods for selecting a function $\hat{f}$---as the output of a linear regression, a lasso regression, and a random forest algorithm---but find that all of these yield similar results.\footnote{The regularization parameter for LASSO is determined using 5-fold cross-validation. For random forest, we utilize a configuration of 300 trees with other hyperparameters set to their default values as specified in scikit-learn version 1.4.1.
}
Figure \ref{fig:ML_improvement} plots the average value of the utilities $U^w$ and $U^b$ (across the iterations of our procedure) under the status quo algorithm, as well as the average values for these utilities given candidate algorithms identified by each of the three methods mentioned above. This figure suggests that the candidate algorithms succeed in improving not only accuracy (each of $U^w$ and $U^b$ increase) but also fairness ($\vert U^w - U^b\vert$ decreases).

\begin{figure}[h]
    \centering
    \includegraphics[height=9cm]{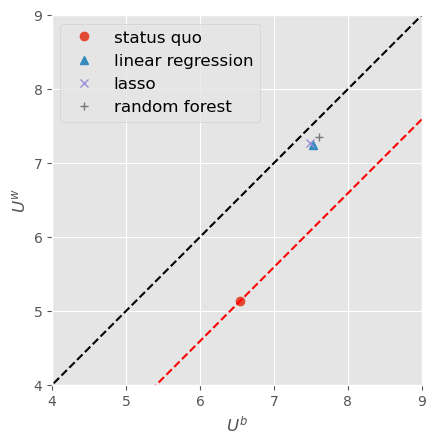}
   \caption{     \footnotesize{We report the average value of $U^g$ (i.e., number of active chronic diseases conditional on automatic enrollment) for each group $g$ across the $K=7$ iterations of our procedure. The algorithm is more accurate for group $g$ when $U^g$ is larger; so, moving towards the upper-right quadrant corresponds to improvements in accuracy. The algorithm is more fair when its $(U^b, U^w)$ pair is closer to the 45-degree line (corresponding to more balanced accuracy across the patient groups). This figure suggests that the candidate algorithms based on linear regression, LASSO, and random forest improve on both fairness and accuracy over the status quo algorithm used by the hospital.}}     \label{fig:ML_improvement}
\end{figure}

We subsequently report findings for the random forest algorithm only, deferring the (nearly identical) results for the other algorithms to Appendix \ref{app:OtherTables}. Table \ref{tab:p_val_rf} complements Figure \ref{fig:ML_improvement}  by reporting the $p$-values that emerge from our statistical test of the null hypothesis that the status quo algorithm is not strictly FA-dominated. Across the iterations of our procedure, the median $p$-value is $0.0001$; thus, we reject the null at the 5\% significance level. (Recall that we reject at a $5\%$ significance level if the median $p$-value is less than $0.025$.) Were the hospital to defend the disparate impact of its algorithm on account of the accuracy goal in (\ref{eq:AccuracyCalibration}), a federal funding agency could reject this business necessity claim with strong statistical guarantees. This finding reinforces the bottom-line takeaway from \citet{Obermeyer2019-te}.\footnote{\citet{Obermeyer2019-te} find that training an algorithm to predict health needs (instead of health costs, as the status quo algorithm does) leads to a more equal proportion of Black and White patients among those who are identified for automatic enrollment.}

\begin{table}[hbt]

    \footnotesize
    \begin{tabular}{
        @{}
        l
        *{3}{c}
        |
        *{3}{c}
        |
        *{3}{c}
        |
        c
        @{}
    }
    \toprule
    & \multicolumn{3}{c}{\bfseries Accuracy (Black)} & \multicolumn{3}{c}{\bfseries Accuracy (White)} & \multicolumn{3}{c}{\bfseries Unfairness} &  \\
    \cmidrule(lr){2-4} \cmidrule(lr){5-7} \cmidrule(lr){8-10} \cmidrule(l){11-11}
         & {$a_1$} & {$a_0$} & {$p_b$} & {$a_1$} & {$a_0$} & {$p_w$} & {$a_1$} & {$a_0$} & {$p_f$} & {$p$} \\
    \midrule
        Iteration 1 & 7.44 & 6.33 & 0.0000 & 7.35 & 5.14 & 0.0000 & 0.09 & 1.19 & 0.0000 & 0.0000 \\
        Iteration 2 & 7.50 & 6.32 & 0.0001 & 7.41 & 5.11 & 0.0000 & 0.09 & 1.20 & 0.0000 & 0.0001 \\
        Iteration 3 & 7.55 & 6.67 & 0.0001 & 7.25 & 5.15 & 0.0000 & 0.30 & 1.52 & 0.0000 & 0.0001 \\
        Iteration 4 & 7.46 & 6.35 & 0.0000 & 7.31 & 5.06 & 0.0000 & 0.15 & 1.28 & 0.0000 & 0.0000 \\
        Iteration 5 & 7.76 & 6.88 & 0.0009 & 7.33 & 5.27 & 0.0000 & 0.43 & 1.61 & 0.0000 & 0.0009 \\
        Iteration 6 & 7.86 & 6.52 & 0.0000 & 7.43 & 5.02 & 0.0000 & 0.43 & 1.51 & 0.0002 & 0.0002 \\
        Iteration 7 & 7.66 & 6.74 & 0.0005 & 7.40 & 5.19 & 0.0000 & 0.26 & 1.55 & 0.0001 & 0.0005 \\
    \bottomrule
    \end{tabular}
    
    \caption{\footnotesize{The candidate algorithm $a_1$ in the table is based on random forests. Reported $p$-values are computed via bootstrap with 10,000 iterations. The median $p$-value is 0.0001.}}
    \label{tab:p_val_rf}
\end{table}

We further explore the size of achievable improvements in accuracy and fairness by testing for $(\delta_a,\delta_a,\delta_f)$-improvability across different values of $\delta_a$ and $\delta_f$. Since our primary interest is achievable reductions in disparate impact, we restrict consideration to positive values of $\delta_f$. However we  allow $\delta_a$ to take either sign: positive values of $\delta_a$ correspond to a test of whether the improvement in fairness can be achieved simultaneously to an improvement in accuracy (by at least $\delta_a$ percent), and negative values of $\delta_a$ correspond to a test of whether the improvement in fairness can be achieved by reducing accuracy (by no more than $\vert\delta_a\vert$ percent).

Figure~\ref{fig:2D_plot_all_health} presents $p$-values for these tests, and in particular emphasizes the 0.025 $p$-value contour, which corresponds to the conventional 5\% significance level. This contour lies above $\delta_a=0$ for all $\delta_f \leq 0.64$, meaning that even with a demanding fairness-improvement standard---namely, a 64\% reduction in disparate impact---it is possible to maintain, and in fact improve, accuracy for all groups.
In the other direction, the contour lies above $\delta_f = 0$ only for $\delta_a \leq 0.09$, indicating that we can reject the null hypothesis when testing for accuracy improvements of up to 9\% while maintaining fairness, but not beyond this level. 

Finally, while we illustrated our proposed technique using the calibration utility function in (\ref{eq:AccuracyCalibration}), the same analysis is possible using other definitions of accuracy and fairness. Appendix~\ref{app:accuracy_cost} reproduces Figure~\ref{fig:2D_plot_all_health} for an alternative definition of accuracy that measures expected health care costs (leaving the fairness definition unchanged). That is, an algorithm improves on another if among those patients who are automatically enrolled, the disparity in the expected number of health conditions among Black and White patients is smaller, and moreover, the expected health costs are higher for both Black and White patients. For this specification, we find that we cannot reject the null hypothesis that the hospital's algorithm is FA-dominated. 

\begin{figure}
    \centering
    \includegraphics[width=0.7\linewidth]{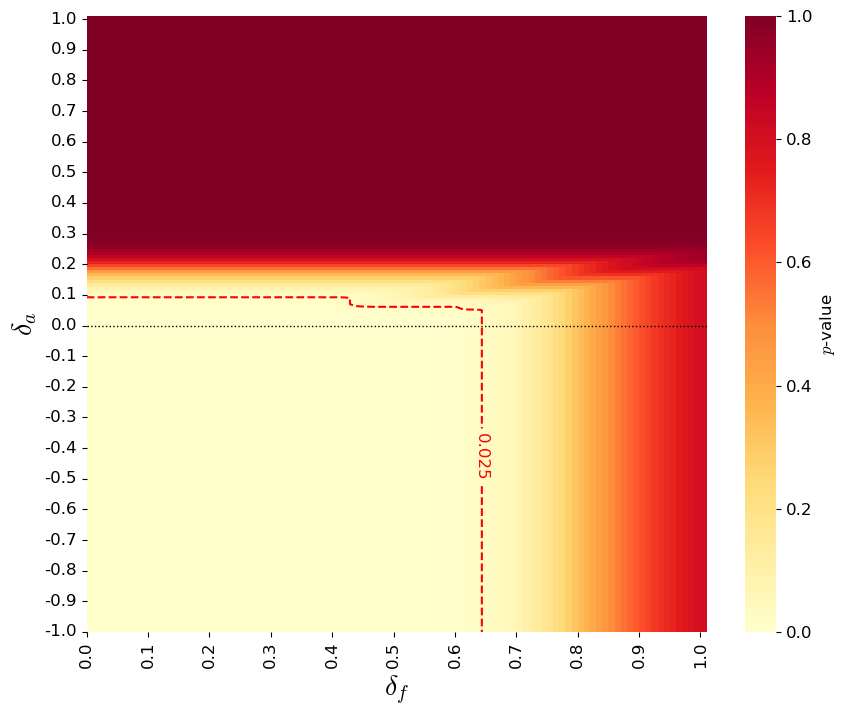}
    \caption{\footnotesize
    We present $p$-values for testing $(\delta_a, \delta_a, \delta_f)$-improvability for $(\delta_a, \delta_f) \in [-1,1] \times [0,1]$. The selection rule is based on random forests, with our procedure setting $K = 7$.
    Both accuracy and fairness utilities measure the expected health needs of those selected into the program.
    Larger values of $\delta_a$ and $\delta_f$ indicate a stricter test regarding the dimensions of fairness and accuracy, respectively. The horizontal line at $\delta_a = 0$ crosses the $p=0.025$ contour at $\delta_f = 0.64$.}
    
    \label{fig:2D_plot_all_health}
\end{figure}

\section{Conclusion}

When a commercial algorithm has a disparate impact, it is important to understand whether it is possible to reduce that disparate impact without compromising on other business-relevant criteria, or if this level of disparate impact is necessitated by those other goals. We have designed a statistical approach to assess this, with three practical objectives in mind. First, since the appropriate definitions of disparate impact and business-relevant criteria vary substantially across applications, we want our framework to be flexible enough to accommodate any such definitions that may emerge in practice. Second, since there are often exogenous constraints on the algorithm space, we want our procedure to apply universally across algorithm classes. Finally, we would like for the approach to be transparent and simple to use for practitioners. The statistical framework that we propose delivers on these three counts. 

The main drawback of our approach is its dependence on a selection rule for choosing a candidate algorithm. When we are able to reject the null, as in our application, then the optimality properties of the selection rule (in particular, whether it satisfies improvement convergence as defined in Section \ref{sec:Full_test}) do not matter. That is, no matter how naive or heuristic our selection rule is, we can conclude that the status quo algorithm is improvable. But when our test does not allow us to reject the null, there is in general an ambiguity: It may be that there is not sufficient evidence in the data to conclude existence of an improving algorithm; or, alternatively, it may be that the selection rule is not powerful enough to find this algorithm. In the latter case, re-running the same test with a different selection rule could lead us to reject the null. This ambiguity is resolved asymptotically when the selection rule is improvement-convergent, since Theorem \ref{thm:consistent_test} establishes that our proposed test is consistent under this condition. One interesting direction for future work is thus to provide sufficient conditions for a selection rule to be improvement-convergent in different  applications. 

\clearpage

\appendix

\section{Bootstrap algorithm }\label{sec:boot}
Let $\hat{P}_n$ be the empirical distribution of the data $\{(X_i, Y_i, G_i)\}_{i = 1}^{\ell_n}$ and let $\{(X^{*}_i, Y^*_i, G^*_i)\}_{i=1}^{\ell_n}$ denote a bootstrap sample drawn i.i.d from $\{(X_i, Y_i, G_i)\}_{i = 1}^{\ell_n}$. Let $\hat{\theta}^*$, $\hat{A}^*_{tg}$, $\hat{F}^*_{tg}$ denote the bootstrap analogs of $\hat{\theta}$, $\hat{A}_{tg}$, $\hat{F}_{tg}$ computed using the bootstrap data. The bootstrap analogs of the test-statistics are given by 
\[\hat{T}_{g,n}^{*} = \hat{A}^*_{1g} - (1 + \Delta_{g})\hat{A}_{0g}^{*}  \text{ for $g \in \{r, b\}$}~,\] 
\[\hat{T}_{f,n}^{*} = \big|\hat{F}_{1r}^{*} - \hat{F}_{1b}^{*}\big| - (1 - \Delta_f)\big|\hat{F}_{0r}^{*} - \hat{F}_{0b}^{*}\big|~.\]

The critical values for the test are based on the quantiles of the bootstrap distribution of the test statistics. Specifically, let $\Psi_{g,n}$ be the (conditional) cumulative distribution function of $\sqrt{\ell_n}(\hat{T}_{g,n}^{*} - \hat{T}_{g,n})$ given $\hat{P}_n$ and $\Psi_{f,n}$ be the (conditional) cumulative distribution function of $\sqrt{\ell_n}(\hat{T}_{f,n}^{*} - \hat{T}_{f,n})$ given $\hat{P}_n$. In practice these are not known exactly, but can be approximated via Monte-Carlo by taking $Q$ independent bootstrap samples from $\hat{P}_n$ and computing the empirical distribution
\[\frac{1}{Q}\sum_{q=1}^Q I\{\sqrt{\ell_n}(\hat{T}^*_{s,n,q} - \hat{T}_{s,n}) \le x\}~,\]
for $s \in \{r, b, f\}$, where $\hat{T}^*_{s,n,q}$ denotes the bootstrap test statistic computed on the $q$th sample.

\section{Preliminary Results}

Below we establish a sequence of preliminary results, which will be used to prove our Theorems \ref{thm:full_test} and \ref{thm:consistent_test}. Section \ref{app:ReducetoComponent} proves a result that allows us to reduce asymptotic validity of our joint test to asymptotic validity of tests of the component nulls. Section \ref{sec:boot_consist} proves bootstrap consistency in a simpler environment where the sequence of candidate algorithms $a_{1n}$ is deterministic.

\subsection{Intersection-Union Testing} \label{app:ReducetoComponent}

\begin{proposition}\label{prop:IUT}
Let $\Omega$ denote a set of distributions for which $P \in \Omega$. Let $\Omega_0 \subset \Omega$ denote the subset of distributions for which $H_0$ holds. Let $\phi^{(r)}_n \in \{0, 1\}$, $\phi^{(b)}_n \in \{0,1\}$, $\phi^{(f)}_n \in \{0, 1\}$ be tests of the component null hypotheses for accuracy and fairness as defined in Section \ref{sec:Test}.
Let $\Omega_{0s} \subset \Omega$ for $s\in \{r,b,f\}$ denote the subsets of distributions for which each component null hypothesis $H_{0s}$ holds. Define the test $\phi_n = \phi^{(r)}_n\cdot\phi^{(b)}_n\cdot\phi^{(f)}_n$ which rejects if and only if each test $\phi^{(s)}_n$ rejects for $s \in \{r, b, f\}$. 
\begin{enumerate}
\item If $\phi^{(s)}_n$ for $s \in \{r, b, f\}$ are asymptotically of level $\alpha$, i.e., if
\[\limsup_{n \rightarrow \infty}E_P[\phi_n^{(s)}] \le \alpha~,\]
for every $P \in \Omega_{0s}$, then $\phi_n$ is asymptotically of level $\alpha$, i.e. 
\[\limsup_{n \rightarrow \infty} E_P[\phi_n] \le \alpha~,\]
for every $P \in \Omega_0$.
\item If $\phi^{(s)}_n$ for $s \in \{r, b, f\}$ are uniformly asymptotically of level $\alpha$, i.e. if
\[\limsup_{n \rightarrow \infty}\sup_{P \in \Omega_{0s}}E_P[\phi_n^{(s)}] \le \alpha~,\]
then $\phi_n$ is uniformly asymptotically of level $\alpha$, i.e.
\[\limsup_{n \rightarrow \infty}\sup_{P \in \Omega_{0}}E_P[\phi_n] \le \alpha~.\]
Moreover, if there exists some distribution $P^* \in \Omega_0$ such that $E_{P^*}[\phi_n^{(s)}] \rightarrow \alpha$ for some $s \in \{r, b,  f\}$, and $E_{P^*}[\phi_n^{(s')}] \rightarrow 1$ for $s' \ne s$, then
\[\lim_{n \rightarrow \infty}\sup_{P \in \Omega_{0}}E_P[\phi_n] = \alpha~.\]
\end{enumerate}
\end{proposition}
\begin{proof}
To prove the first claim, note that by the definition of $\phi_n$, for any $P \in \Omega_0$, 
\[E_P[\phi_n] \le E_P[\phi^{(s)}_n]~,\]
for every $s \in \{r, b, f\}$. By the definition of $\Omega_0$, $P \in \Omega_{0s}$ for some $s \in \{r, b, f\}$. Since $\phi^{(s)}_n$ is asymptotically of level $\alpha$, it follows that 
\[\limsup_{n \rightarrow \infty}E_P[\phi_n] \le \limsup_{n \rightarrow \infty}E_P[\phi^{(s)}_n] \le \alpha~,\]
as desired. 

To prove the second claim, first note that by following the argument above, for each $P \in \Omega_0$ we obtain 
\[E_P[\phi_n] \le \sup_{P' \in \Omega_{0s}}E_{P'}[\phi^{(s)}_n]~,\]
for some $s \in \{r, b, f\}$. It thus follows that 
\[\sup_{P \in \Omega_0}E_P[\phi_n] \le \max\left\{ \sup_{P' \in \Omega_{0s}}E_{P'}[\phi^{(s)}_n]: s \in \{r, b, f\}\right\}~.\]
Since each $\phi^{(s)}_n$ is uniformly asymptotically of level $\alpha$, we then obtain that 
\begin{equation}\label{eq:limsup}
\limsup_{n \rightarrow \infty}\sup_{P \in \Omega_0}E_P[\phi_n] \le \alpha~.
\end{equation}
Next, note that from the definition of $\phi_n$,
\begin{align*}
\sup_{P \in \Omega_0}E_P[\phi_n] &\ge E_{P^*}[\phi_n] \\
&\ge \sum_{s}E_{P^*}[\phi_n^{(s)}] - 2~.
\end{align*}
It thus follows by our assumptions on $P^*$ that 
\begin{equation}\label{eq:liminf}
\liminf_{n \rightarrow \infty}\sup_{P \in \Omega_0}E_{P}[\phi_n] \ge \alpha~.
\end{equation}
Combining \eqref{eq:limsup} and \eqref{eq:liminf} we thus obtain that 
\[\lim_{n \rightarrow \infty}\sup_{P \in \Omega_0}E_{P}[\phi_n] = \alpha~,\]
as desired.
\end{proof}

\subsection{Bootstrap consistency}\label{sec:boot_consist}
To demonstrate Theorems \ref{thm:full_test} and \ref{thm:consistent_test}, we first show that the bootstrap procedure defined in Appendix \ref{sec:boot} is consistent when the test statistics are applied to an arbitrary deterministic sequence of candidate algorithms $a_{1n} \in \mathcal{A}$, i.e. that
\[d_{\infty}\left(\mathcal{L}\left\{\sqrt{\ell_n}(\hat{T}_{s,n} - T_{s,n})\right\}, \Psi_{s,n}\right) \xrightarrow{P} 0~,\]
as $n \rightarrow \infty$, for $s \in \{r, b, f\}$, where $d_{\infty}(\cdot, \cdot)$ denotes the Kolmogorov metric and $\mathcal{L}\{Z\}$ denotes the distribution function of a random variable $Z$. Note here that we have been careful to index each term in the expression by the total sample size $n$, since in principle each of these objects could vary with the candidate algorithm $a_{1n}$, but we suppress this dependence subsequently whenever we think it would not lead to confusion.

We prove bootstrap consistency in three steps. In the first step (Proposition \ref{prop:lin_normal}), we show that the asymptotic linear representations of the test statistics are asymptotically normal. In the second step (Proposition \ref{prop:infeasible_boot}), we demonstrate consistency of the bootstrap distribution for the asymptotic linear representations by invoking Theorem 1 of \cite{mammen2012does}. Finally in the third step we show that the bootstrap remains consistent when we move from the asymptotic linear representations to the original test statistics (Proposition \ref{prop:new_feasible_boot}). Throughout this section we maintain Assumptions \ref{assp:n}, \ref{assp:Utilities} and \ref{assp:NonDegenerate}.

\subsubsection{Asymptotic normality of the test statistic}
We first write the test statistics as sample averages. To do this, we define 
\begin{align}
\Omega^A_{i} &= (u^A_{1r}(Z_i,\theta), u^A_{0r}(Z_i,\theta),u^A_{1b}(Z_i,\theta),u^A_{0b}(Z_i,\theta))~, \nonumber \\
\hat{\Omega}^A_{i} &= (u^A_{1r}(Z_i,\hat{\theta}), u^A_{0r}(Z_i,\hat{\theta}),u^A_{1b}(Z_i,\hat{\theta}),u^A_{0b}(Z_i,\hat{\theta}))~, \nonumber \\
\Omega^F_{i} &= (u^F_{1r}(Z_i,\theta), u^F_{0r}(Z_i,\theta),u^F_{1b}(Z_i,\theta),u^F_{0b}(Z_i,\theta))~, \nonumber \\
\hat{\Omega}^F_{i} &= (u^F_{1r}(Z_i,\hat{\theta}), u^F_{0r}(Z_i,\hat{\theta}),u^F_{1b}(Z_i, \hat{\theta}),u^F_{0b}(Z_i, \hat{\theta}))~, \nonumber \\
\mu_{i} &= (h_i,\Omega^A_i,\Omega^F_i)~, \nonumber \\
\hat{\mu}_{i} &= (h_{i}, \hat{\Omega}^A_{i},\hat{\Omega}^F_{i})~, \nonumber \\
\Sigma_n &= E\left[(\mu_{i} - E[\mu_i])(\mu_{i} - E[\mu_i])^{T}\right] \label{eq:sigma}
\end{align}
and
\begin{align*}
\alpha_r &= (0_J,1,-(1+\Delta_r),0_6) ~, \\
\alpha_{b} &= (0_{J+2},1,-(1+\Delta_{b}),0_4) ~, \\
\alpha_{f} &=  \begin{pmatrix} 0_{(J+4) \times 4} \\ I_4 \\
\end{pmatrix}~,
\end{align*}
where $0_{m}$ refers to a $m \times 1$ dimensional vector of $0$s, $0_{m \times n}$ refers to a $m \times n$ dimensional matrix of $0$s, and $I_m$ is the $m\times m$ dimensional identity matrix. 
We also define
\begin{align*}
T_{r,n} &= E\left[\alpha_r^{T} \mu_i\right] ~,\\
T_{b,n} &= E\left[\alpha_{b}^{T} \mu_i\right] ~,\\
T_{f,n} &= \varphi(E\left[\alpha_f^{T}\mu_i\right])~,
\end{align*}
where $\varphi(a,b,c,d) = |a-c| - (1-\Delta_f)|b-d|$ and
\begin{align*}
\hat{T}_{r,n} &= \frac{1}{\ell_n}\sum_{i=1}^{\ell_n}\alpha_r^{T} \hat{\mu}_i ~,\\
\hat{T}_{b,n} &= \frac{1}{\ell_n}\sum_{i=1}^{\ell_n}\alpha_{b}^{T} \hat{\mu}_i ~,\\
\hat{T}_{f,n} &= \varphi\left(\frac{1}{\ell_n}\sum_{i=1}^{\ell_n}\alpha_f^{T}\hat{\mu}_i\right)~.
\end{align*}

To demonstrate asymptotic normality of the test statistics, we derive their asymptotic linear representations. To do this, we use the additional definitions
\[A^{1}_{tg} = E[\nabla u^A_{tg}(Z_i,\theta)]~,\]
\[F^{1}_{tg} = E[\nabla u^F_{tg}(Z_i,\theta)]~.\] 
where we note that these are $J \times 1$ dimensional vectors, and
\begin{equation*}
\begin{aligned}
\tilde{\alpha}_{r} &= (A_{1r}^{1} - A_{0r}^{1}(1+\Delta_{r}), 1,-(1+\Delta_r),0_{6})~, \\
\tilde{\alpha}_{b} &= (A_{1b}^{1} - A_{0b}^{1}(1+\Delta_{b}), 0_{2}, 1,-(1+\Delta_{b}),0_{4})~, \\
\tilde{\alpha}_{f} &= \left(\text{sign}(F_{1r}-F_{1b})(F_{1r}^{1}-F_{1b}^{1}) - (1-\Delta_f)\text{sign}(F_{0r}-F_{0b})(F_{0r}^{1}-F_{0b}^{1}), 0_{4},\right. \\
&\hspace{10mm}\left.\text{sign}(F_{1r}-F_{1b}), -(1-\Delta_f)\text{sign}(F_{0r}-F_{0b}), -\text{sign}(F_{1r}-F_{1b}), (1-\Delta_f)\text{sign}(F_{0r}-F_{0b})\right)~.
\end{aligned}
\end{equation*}

 Using these additional constructions, we can define the asymptotic linear representations of the test statistics and demonstrate their asymptotic normality. They are \begin{equation*}
\begin{aligned}
\tilde{T}_{r,n} &= T_{r,n} + \frac{1}{\ell_n}\sum_{i=1}^{\ell_n}\tilde{\alpha}_r^{T}(\mu_{i} - E[\mu_i]), \\
\tilde{T}_{b,n} &= T_{b,n} + \frac{1}{\ell_n}\sum_{i=1}^{\ell_n}\tilde{\alpha}_b^{T}(\mu_{i} - E[\mu_i]), \\
\tilde{T}_{f,n} &=  T_{f,n} + \frac{1}{\ell_n}\sum_{i=1}^{\ell_n}\tilde{\alpha}_{f}^{T}(\mu_{i} - E[\mu_i]))~.
\end{aligned}
\end{equation*} 

First, we argue that the asymptotic linear representations are asymptotically normal.

\begin{proposition}\label{prop:lin_normal}
For $g \in \{r,b\}$,
\[d_{\infty}\left(\mathcal{L}\left\{\sqrt{\ell_n}(\tilde{T}_{g,n} - T_{g,n})\right\},\mathcal{N}(0,\sigma_{g,n}^{2})\right) \rightarrow 0~,\] and 
\[d_{\infty}\left(\mathcal{L}\left\{\sqrt{\ell_n}(\tilde{T}_{f,n} - T_{f,n})\right\}, \mathcal{N}(0,\sigma_{f,n}^{2})\right) \rightarrow 0~,\] where $\sigma_{g,n}^2 = \tilde{\alpha}_g^T \Sigma_n \tilde{\alpha}_g$ and $\sigma_{f,n}^2 = \tilde{\alpha}_f^T \Sigma_n \tilde{\alpha}_f$ for $g \in \{r,b\}$, and $d_\infty(\cdot, \cdot)$ is the Kolmogorov metric.
\end{proposition}
\begin{proof} 
For $g \in \{r,b\}$
 \[\sqrt{\ell_n}(\tilde{T}_{g,n} - T_{g,n}) 
 = \frac{1}{\sqrt{\ell_n}}\sum_{i}\tilde{\alpha}_{g}^{T}(\mu_i - E[\mu_i])~.\]
Since $\mu_i$ is independent across $i$, the Lindeberg-Feller CLT implies that 
 \[\frac{\sqrt{\ell_n}(\tilde{T}_{g,n} - T_{g,n})}{\sigma_{g,n}} \xrightarrow{d} \mathcal{N}(0,1)~,\]
 where $\sigma_{g,n}^2 = \tilde{\alpha}_g^T \Sigma_n \tilde{\alpha}_g$. Note that $\sigma_{g,n}^2 >0$ by Assumption \ref{assp:NonDegenerate}. It thus follows by Polya's Theorem \citep[Theorem 11.2.9 in][]{lehmann2005testing} that
 \[d_{\infty}\left(\mathcal{L}\left\{\sqrt{\ell_n}(\tilde{T}_{g,n} - T_{g,n})\right\},\mathcal{N}(0,\sigma_{g,n}^{2})\right) \rightarrow 0~.\]

 Similarly, 
\[
\sqrt{\ell_n}(\tilde{T}_{f,n} - T_{f,n}) 
= \frac{1}{\sqrt{\ell_n}} \sum_{i} \tilde{\alpha}_{f}^{T} (\mu_i - E[\mu_i])~,
\]
and
 \[d_{\infty}\left(\mathcal{L}\left\{\sqrt{\ell_n}(\tilde{T}_{f,n} - T_{f,n})\right\}, \mathcal{N}(0,\sigma_{f,n}^{2})\right) \rightarrow 0~,\]
  where $\sigma_{f,n}^2 = \tilde{\alpha}_f^T \Sigma_n \tilde{\alpha}_f$. Note also that $\sigma_{f,n}^2 > 0$ by Assumption 2. 
\end{proof}

Second, we argue that the difference between the test statistics and their asymptotic linear representation are asymptotically negligible. 

 \begin{proposition}\label{prop:lin_equiv} Suppose there exists some $\xi > 0$ such that $|U_F^r(a_{1n}) - U_F^{b}(a_{1n})| \ge \xi$ for every $a_{1n}$ in our sequence of candidate algorithms. Then for each $g \in \{r,b\}$,
      \[\sqrt{\ell_n}(\hat{T}_{g,n} - \tilde{T}_{g,n}) \xrightarrow{P} 0 \quad \text{ and } \quad \sqrt{\ell_n}(\hat{T}_{f,n} - \tilde{T}_{f,n}) \xrightarrow{P} 0~.\]
 \end{proposition}
\begin{proof}
The result for the accuracy test statistics follows from
\begin{align*}
\hat{T}_{g,n} 
&= \alpha_g^T \frac{1}{\ell_n}\sum_i \hat{\mu}_i
= \alpha_g^T \frac{1}{\ell_n}\sum_i (\mu_i + \nabla \mu_i (\hat{\theta} - \theta) ) + o_{p}(\ell_n^{-1/2}) \\
&= \tilde{\alpha}^{T}_g \frac{1}{\ell_n}\sum_i(0_J,\Omega_i^A, \Omega_i^F) + \tilde{\alpha}^{T}_{g,n} \frac{1}{\ell_n}\sum_i(h_i - \theta,0_8) + o_{p}(\ell_n^{-1/2}) \\
&= \frac{1}{\ell_n}\sum_i \tilde{\alpha}_g^{T}(\mu_i - (\theta,0_8)) + o_{p}(\ell_n^{-1/2}) = \tilde{T}_{g,n} + o_{p}(\ell_n^{-1/2})~,
\end{align*}
where the last equality is because $\tilde{\alpha}_g^{T}(\theta,0_{8}) + \alpha^{T}_{g}E[\mu_i] = \tilde{\alpha}_g^{T}E[\mu_i]$.

The second equality follows from a Taylor expansion, Assumption \ref{assp:Utilities}, and the fact that $||\hat{\theta} - \theta|| = O_{p}(\ell_n^{-1/2})$ by a standard concentration argument (e.g. Markov's inequality).  The third equality follows from a law of large numbers, $||\hat{\theta} - \theta|| = O_{p}(\ell_n^{-1/2})$, and some additional algebra.

The result for the fairness test statistic follows from 
\begin{align*}
\hat{T}_{f,n} 
&= \varphi\left(\frac{1}{\ell_n}\sum_{i}\alpha_f^{T}\hat{\mu}_i \right) 
= \varphi(E[\alpha_f^T \mu_i]) + \nabla \varphi(E[\alpha_f^T \mu_i])\left(\frac{1}{\ell_n}\sum_i \alpha_f^T\hat{\mu}_i - E[\alpha_f^T \mu_i]\right) + o_{p}(\ell_n^{-1/2}) \\
&= \varphi(E[\alpha_f^T \mu_i]) + \nabla\varphi(E[\alpha_f^T \mu_i])\alpha_f^T\left(\frac{1}{\ell_n}\sum_i(\mu_i - E[\mu_i]) + E[\nabla\mu_i](\hat{\theta} - \theta)\right) + o_{p}(\ell_n^{-1/2}) \\
&= \varphi(E[\alpha_f^T \mu_i]) + \nabla\varphi(E[\alpha_f^T \mu_i])\alpha_f^T(I_{J+8} + (E[\nabla\mu_i],0_{(J+8)\times 8}))\left(\frac{1}{\ell_n}\sum_i(\mu_i - E[\mu_i]) \right) + o_{p}(\ell_n^{-1/2}) \\
&= T_{f,n} + \frac{1}{\ell_n}\sum_i \tilde{\alpha}_f^T (\mu_i - E[\mu_i]) + o_{p}(\ell_n^{-1/2})
= \tilde{T}_{f,n} + o_{p}(\ell_n^{-1/2})~,
\end{align*}
where $I_{J+8}$ is a $(J+8)\times (J+8)$ dimensional identity matrix and $0_{(J+8)\times 8}$ is a $(J+8) \times 8$ dimensional matrix of $0$s. Existence and continuity of $\nabla\varphi(\cdot)$ follows by Assumption \ref{assp:n} and the assumption that $|U_F^r(a_{1n}) - U_F^{b}(a_{1n})| \ge \xi$.

The second equality follows from a Taylor expansion and because $\frac{1}{\ell_n}\sum_i \alpha_f^T\left(\hat{\mu}_i - E[\mu_i]\right) = O_{p}(\ell_n^{-1/2})$. The third equality is because $\frac{1}{\ell_n}\sum_i \left(\nabla \mu_i - E[\nabla \mu_i]\right)= O_{p}(\ell_n^{-1/2})$ and $||\hat{\theta} - \theta|| = O_{p}(\ell_n^{-1/2})$. They both follow from Assumption 1 and a standard concentration argument.

The second-to-last equality is because $T_{f,n} = \varphi(E[\alpha_f^T \mu_i])$ and $\tilde{\alpha}_f^T = \nabla\varphi(E[\alpha_f^T \mu_i])\alpha_f^T(I_{J+8} + (E[\nabla\mu_i],0_{(J+8)\times 8}))$. This last result is because
\[\nabla\varphi(E[\alpha^T_f \mu_i]) = (\text{sign}(F_{1r}-F_{1b}),-(1-\Delta_f)\text{sign}(F_{0r}-F_{0b}),-\text{sign}(F_{1r}-F_{1b}),-(1-\Delta_f)\text{sign}(F_{0r}-F_{0b}))~,\]
\[\alpha_f^T = \begin{pmatrix} 0_{4 \times (J+4)} & I_4 \end{pmatrix}~, \text{ and }\]
\[\left(I_{J+8} + (E[\nabla \mu_i],0_{(J+8)\times 8})\right) 
= \begin{pmatrix} I_{J} & 0 & 0 \\
A^{1} & I_{4} & 0 \\
F^{1} & 0 & I_{4} \\
\end{pmatrix}\]
for $A^{1} = (A^{1}_{1r},A^{1}_{0r},A^{1}_{1b},A^{1}_{0b})$ and $F^{1} = (F^{1}_{1r},F^{1}_{0r},F^{1}_{1b},F^{1}_{0b})$
so that
\begin{align*}
&\nabla\varphi(E[\alpha_f^T \mu_i])\alpha_f^T(I_{J+8} + (E[\nabla\mu_i],0_{(J+8)\times 8})) \\
&= \left(\text{sign}(F_{1r}-F_{1b})(F_{1r}^1-F_{1b}^1)-(1-\Delta_f)\text{sign}(F_{0r} - F_{0b})(F_{0r}^1 - F_{0b}^1), 0_4, \right. \\
&\hspace{10mm} \left. \text{sign}(F_{1r}-F_{1b}), -(1-\Delta_f)\text{sign}(F_{0r}-F_{0b}), -\text{sign}(F_{1r}-F_{1b}), -(1-\Delta_f)\text{sign}(F_{0r}-F_{0b}) \right) \\
&= \tilde{\alpha}_f~.
\end{align*}
\end{proof}

It follows from Propositions \ref{prop:lin_normal} and \ref{prop:lin_equiv} that for $g \in \{r,b\}$, 
\[d_{\infty}\left(\mathcal{L}\left\{\sqrt{\ell_n}(\hat{T}_{g,n} - T_{g,n})\right\},\mathcal{N}(0,\sigma_{g,n}^{2})\right) \rightarrow 0~,\] and 
\[d_{\infty}\left(\mathcal{L}\left\{\sqrt{\ell_n}(\hat{T}_{f,n} - T_{f,n})\right\}, \mathcal{N}(0,\sigma_{f,n}^{2})\right) \rightarrow 0~.\] 

\subsubsection{Asymptotic normality of the bootstrap distribution}
Let the (infeasible) bootstrap distributions of $\sqrt{\ell_n}(\tilde{T}_{g,n} - T_{g,n})$ and $\sqrt{\ell_n}(\tilde{T}_{f,n} - T_{f,n})$ be denoted by $\tilde{\Psi}_{g,n}$ and $\tilde{\Psi}_{f,n}$. We say that these distributions are infeasible because they depend on unknown population parameters. 

\begin{proposition}\label{prop:infeasible_boot}
$d_\infty( \tilde{\Psi}_{g,n}, \mathcal{N}(0,\sigma^2_{g,n})) \xrightarrow{P} 0$ and $d_\infty( \tilde{\Psi}_{f,n}, \mathcal{N}(0,\sigma^2_{f,n})) \xrightarrow{P} 0$.
\end{proposition}
\begin{proof}
This follows from Proposition \ref{prop:lin_normal} above and Theorem 1 of \cite{mammen2012does}.
\end{proof}

\begin{proposition}\label{prop:new_feasible_boot}
Suppose $|U_F^r(a_{1n}) - U_F^{b}(a_{1n})| \ge \xi$ for some $\xi > 0$, for every $a_{1n}$ in our sequence of candidate algorithms. For $s \in \{r, b ,f\}$,
\[d_{\infty}\left(\mathcal{L}\left\{\sqrt{\ell_n}(\hat{T}_{s,n} - T_{s,n})\right\}, \Psi_{s,n}\right) \xrightarrow{P} 0~,\]
as $n \rightarrow \infty$.
\end{proposition}
\begin{proof}
Let $P_*$ denote the bootstrap distribution conditional on the data. By applying the logic of Proposition \ref{prop:lin_equiv} to the bootstrap data, it follows that
\[\sqrt{\ell_n}(\hat{T}^*_{s,n} - \tilde{T}^*_{s,n}) = r^*_{s,n}~,\]
for some $r^*_{s,n}$ such that $P_*\left(|r^*_{s,n}| > \epsilon\right) \xrightarrow{P} 0$. 
To establish our result, we argue along subsequences. To that end, for every subsequence there exists a further subsequence (which we index by $n_k$) such that
\[
\begin{aligned}
&P_*\left(|r^*_{s,n_k}| > \epsilon\right) \xrightarrow{a.s.} 0~, \\
&\sqrt{\ell_n}(\hat{T}_{s,n_k} - \tilde{T}_{s,n_k}) \xrightarrow{a.s.} 0~, \\
&d_{\infty}(\tilde{\Psi}_{s,n_k}, \mathcal{N}(0, \sigma^2_{s,n_k})) \xrightarrow{a.s.} 0 ~, \\
&\sigma_{s,n_k} \rightarrow \sigma^{\dagger} > 0~,
\end{aligned}
\]
where the second convergence statement follows from Proposition \ref{prop:lin_equiv}, the third from Proposition \ref{prop:infeasible_boot}, and the last convergence statement follows from Assumptions \ref{assp:Utilities} and \ref{assp:NonDegenerate}.
We thus have that 
\[\frac{\sqrt{\ell_{n_k}}(\hat{T}^*_{s,n_k} - \hat{T}_{s,n_k})}{\sigma_{s,n_k}} = \frac{\sqrt{\ell_{n_k}}(\hat{T}^*_{s,n_k} - \tilde{T}^*_{s,n_k})}{\sigma_{s,n_k}} + \frac{\sqrt{\ell_{n_k}}(\tilde{T}^*_{s,n_k} - \tilde{T}_{s,n_k})}{\sigma_{s,n_k}} + \frac{\sqrt{\ell_{n_k}}(\tilde{T}_{s,n_k} - \hat{T}_{s,n_k})}{\sigma_{s,n_k}} \\
\xrightarrow{d} N(0, 1)~,\]
conditional on the data, with probability one.
It follows by Polya's Theorem \citep[Theorem 11.2.9 in][]{lehmann2005testing}  that
\[d_{\infty}\left(\Psi_{s,n_k}, N(0,\sigma^2_{s,n_k})\right) \xrightarrow{a.s.} 0~,\]
and thus by Propositions \ref{prop:lin_normal} and \ref{prop:lin_equiv},
\[d_{\infty}\left(\mathcal{L}\left\{\sqrt{\ell_{n_k}}(\hat{T}_{s,n_k} - T_{s,n_k})\right\}, \Psi_{s,n_k}\right) \xrightarrow{a.s.} 0~.\]
Since we have established that for every subsequence there exists a further subsequence along which the convergence holds almost surely, the original sequence converges in probability, as desired.
\end{proof}

\section{Proofs of Results in Section \ref{sec:Full_test}}
\subsection{Proof of Theorem \ref{thm:full_test}}

We will first show that for any  deterministic sequence of algorithms $a_{1n} \in \mathcal{A}$, and any $P$ satisfying the null hypothesis,
\begin{equation} \label{eq:ConvDeterministic} \limsup_{n \rightarrow \infty}E_P[\phi_n(a_0, a_{1n})] \le \alpha~.
\end{equation}
By Proposition \ref{prop:IUT} (1) it suffices to check that each test $\phi_n^{(s)}$ is asymptotically of level $\alpha$ for distributions $P$ contained in $\Omega_{0s}$, $s \in \{r, b, f\}$.
We prove the result for the fairness test since the other results follow from similar arguments. By definition,
\begin{align}
E_P[\phi_n^{(f)}] &= P(\sqrt{\ell_n}\hat{T}_{f,n} < c^*_\alpha) \nonumber \\
&\le P(\sqrt{\ell_n}(\hat{T}_{f,n} - T_{f,n}) < c^*_\alpha)~, \label{ineq:calpha}
\end{align}
since $T_{f,n} \ge 0$ for $P \in \Omega_{0f}$. 

Further note by Assumption \ref{assp:n} that for any $P \in \Omega_{0f}$, \[|U_F^r(a_{1n}) - U_F^{b}(a_{1n})| \ge |F_{0r} - F_{0b}|(1 - \Delta_f) > 0\] for every $a_{1n}$ in our sequence of candidate algorithms. Thus by Proposition \ref{prop:new_feasible_boot}, for any $\epsilon > 0$ and $n$ sufficiently large,
\[P\left\{d_{\infty}\left(\mathcal{L}\left\{\sqrt{\ell_n}(\hat{T}_{f,n} - T_{f,n})\right\}, \Psi_{f,n}\right) \le \frac{\epsilon}{2}\right\} \ge 1 - \frac{\epsilon}{2}~\]
where $d_{\infty}(\cdot, \cdot)$ denotes the Kolmogorov metric and $\mathcal{L}\{Z\}$ denotes the distribution function of a random variable $Z$. Thus by Lemma A.1. (vii) in \cite{romano2012uniform}, we have that for $n$ sufficiently large,
\[P(\sqrt{\ell_n}(\hat{T}_{f,n} - T_{f,n}) \ge c^*_\alpha) \ge 1 - \alpha - \epsilon~.\]
Combining this with (\ref{ineq:calpha}), it follows that for $n$ sufficiently large,
\[P(\sqrt{\ell_n}\hat{T}_{f,n} < c^*_\alpha) \le 1 - (1-\alpha - \epsilon) = \alpha + \epsilon~,\]
from which we can conclude that 
\[\limsup_{n \rightarrow \infty} E_P[\phi_n^{(f)}] \le \alpha~.\]
Repeating this argument for the other two component nulls, we obtain the desired (\ref{eq:ConvDeterministic}).

Next we argue that the result holds for the random sequence $\hat{a}^\rho_{1n}$. By the law of iterated expectations, 
\[E_P[\phi_n(a_0,\hat{a}^{\rho}_{1n})] = E\left[E_P[\phi_n(a_0, \hat{a}^\rho_{1n}) \mid S^{n}_{train}]\right]~. \]
By the independence of $S^{n}_{train}$ from $S^{n}_{test}$,  $E_P[\phi_n(a_0, \hat{a}^\rho_{1n}) \mid S^{n}_{train}] = h(\hat{a}^\rho_{1n})$, where $h(a) = E_P[\phi_n(a_0,a)]$. Thus by Fatou's lemma and (\ref{eq:ConvDeterministic}), we obtain
\[\limsup_{n \rightarrow \infty}E_P[\phi_n(a_0,\hat{a}^{\rho}_{1n})] \le  E\left[\limsup_{n \rightarrow \infty}h(\hat{a}^\rho_{1n})\right] \le \alpha~,\]
as desired.

\subsection{Proof of Corollary \ref{cor:split_test}}
By Markov's inequality
\[E_P[\psi_n] \le \frac{2}{K}\sum_{k = 1}^KE_P[\phi_{n,k}(a_0, \hat{a}^{\rho}_{1nk})]~.\]
The result then follows immediately from Theorem \ref{thm:full_test}.

\subsection{Proof of Theorem \ref{thm:consistent_test}}

Consider a deterministic sequence of algorithms $a_{1n} \in \mathcal{A}$ satisfying:
\begin{enumerate}
    \item there exists a $\xi > 0$ such that $|U_F^r(a_{1n}) - U_F^{b}(a_{1n})| \ge \xi$ for each $a_{1n}$, and
    \item  $U^g_h(a_{1n}) \rightarrow U^g_h(\gamma)$ for each $g \in \{r, b\}$ and $h \in \{A, F\}$, where  $\gamma \in \mathcal{A}$ is an algorithm that $\delta$-fairness improves on the status quo algorithm $a_0$.
\end{enumerate}

We first argue that 
$E_P[\phi_n^{(s)}] \rightarrow 1$ for $s \in \{r, b, f\}$. We show the result for $\phi_n^{(g)}$ for $g \in \{r,b\}$, as the result for fairness follows similarly. We proceed along subsequences. By Assumptions \ref{assp:Utilities} and \ref{assp:NonDegenerate}, for every subsequence there exists a further subsequence (which we index by $n_k$) such that $\sigma^2_{g,n_k} \rightarrow \sigma^2$ for some $\sigma^2 > 0$. Thus by Proposition \ref{prop:new_feasible_boot}
\begin{align}\label{eq:boot_normal}
d_{\infty}\left(\Psi_{g,n_k}, N(0,\sigma^2)\right) \xrightarrow{P} 0~.
\end{align}
Along a further subsequence (which we continue to index by $n_k$), we obtain that \eqref{eq:boot_normal} holds almost surely, and thus by Lemma 11.2.1 in \cite{lehmann2005testing}, $c^*_{1 - \alpha} \xrightarrow{a.s} c_{1 - \alpha}$ along this subsequence, where $c_{1- \alpha}$ denotes the $1 - \alpha$ quantile of a $N(0, \sigma^2)$ distribution. We thus obtain that, along this subsequence,
\begin{align*}
    E_P[\phi_{n_k}^{(g)}] &= P\left(\sqrt{\ell_{n_k}}\hat{T}_{g,n_k} > c^*_{1 - \alpha}\right) \\
    &= P\left(\sqrt{\ell_{n_k}}\left(\hat{T}_{g,n_k} - T_{g,n_k}\right) > c^*_{1 - \alpha} - \sqrt{\ell_{n_k}}T_{g,n_k}\right) \rightarrow 1~,
\end{align*}
where the convergence follows from Corollary 11.3.1 in \cite{lehmann2005testing} and the fact that, since $\gamma$ is a fairness-improvement,
\[c^*_{1 - \alpha} - \sqrt{\ell_{n_k}}T_{g,n_k} \xrightarrow{a.s} -\infty~.\]
Since we have established that for every subsequence there exists a further subsequence along which $E_P[\phi_{n_k}^{(g)}] \rightarrow 1$, we obtain $E_P[\phi_n^{(g)}] \rightarrow 1$. Similar arguments establish $E_P[\phi_n^{(f)}] \rightarrow 1$. Next, we argue that these results continue to hold when we replace the sequence $a_{1n}$ with $\hat{a}^{\rho}_{1n}$. To that end, once again we argue along subsequences. For each subsequence $\phi_{n_k}^{(g)}$, there exists a further subsequence such that $U^g_h(\hat{a}^{\rho}_{1n}) \xrightarrow{a.s} U^g_h(\gamma)$ for each $g \in \{r, b\}$, $h \in \{A,F\}$. Then by the law of iterated expectations and the dominated convergence theorem,
\[E_P[\phi_{n_k}^{(g)}] = E[E[\phi_{n_k}^{(g)}|S^{n}_{train}]] \rightarrow 1~,\]
so that the result follows along the original sequence. Finally we establish the main result of the theorem: by construction
\[E_P[\phi_n(a_0, \hat{a}^{\rho}_{1n})] \ge E_P[\phi_n^{(r)}] +  E_P[\phi_n^{(b)}] +  E_P[\phi_n^{(f)}] - 2 \rightarrow 1~,\]
as desired.

\section{Proofs of Results in Section \ref{sec:RobustManipulation}}

\subsection{Proof of Theorem \ref{prop:RobustManipulationGeneral}}

 The plan of the proof is to demonstrate that there exists an $\varepsilon>0$ such that
  \begin{equation}
 P\left(\max_{1\leq i \leq m_1^*} \phi_{i1}^\alpha = 1\right)  > \alpha + \varepsilon > P\left(\max_{1\leq i \leq m_2^*} \mathbbm{1}\left(\frac1K \sum_{k=1}^K \phi_{ik}^{\alpha/2}\geq \frac12 \right)=1\right)
 \end{equation} 
 when $K$ is sufficiently large. To show this, we will first show the stronger claim that
 \begin{equation}
 P\left(\max_{1\leq i \leq m } \phi_{i1}^\alpha = 1\right)  > \alpha + \varepsilon
 \end{equation} 
 for every $m>1$, and also that 
 \begin{equation}
 P\left(\max_{1\leq i \leq m } \mathbbm{1}\left(\frac1K \sum_{k=1}^K \phi_{ik}^{\alpha/2}\geq \frac12 \right)=1\right) < \alpha + \varepsilon
 \end{equation}
 for $K$ sufficiently large and $m$ not too large (but including $m=m_2^*$). Since $\alpha + \varepsilon$ is a uniform lower bound on the probability of false rejection under the first procedure (across the number of repetitions $m_1$) and a uniform upper bound on the probability of false rejection under the second procedure (across the number of repetitions $m_2$), we obtain the desired result.

 \begin{lemma} \label{lemm:Part2} For every $\varepsilon>0$, there exist finite $\overline{m}$ and $\overline{K}$ such that 
 \begin{equation} \label{eq:Inequality2}
 P\left(\max_{1\leq i \leq m } \mathbbm{1}\left(\frac1K \sum_{k=1}^K \phi_{ik}^{\alpha/2}\geq \frac12 \right)=1\right) < \alpha + \varepsilon
 \end{equation}
 for every $1 \leq m \leq \overline{m}$ and $K > \overline{K}$. Moreover, $m_2^* < \overline{m}$.
     
 \end{lemma}

\begin{proof} Fix any $m\geq 1$. First consider the LHS of (\ref{eq:Inequality2}), and rewrite this probability as
 \begin{equation} \label{eq:Complement}
 P\left(\max_{1\leq i \leq m} \mathbbm{1}\left(\frac1K \sum_{k=1}^K \phi_{ik}^{\alpha/2}\geq \frac12 \right)=1\right) = 1 - P\left(\frac1K \sum_{k=1}^K \phi_{ik}^{\alpha/2} < \frac12 \quad \forall i\right)
 \end{equation}
Since the infinite array of tests $(\phi^{\alpha/2}_{ik})$ has exchangeable entries, we can apply de Finetti's theorem to express the variables $(\phi_{ik}^{\alpha/2})$ as an i.i.d array $(Z_{ik}^q)$ where each $Z_{ik}^q \sim \text{Ber}(q)$ while $q \sim \mu$. Since $\mathbb{E}(\phi_{ik}^{\alpha/2}) = \alpha/2$, also $\mathbb{E}_\mu(q) = \alpha/2$. Then
 \begin{align*}
 1 - P\left(\frac1K \sum_{k=1}^K \phi_{ik}^{\alpha/2} < \frac12 \quad \forall i\right) & = 1 - \int_{q \in [0,1]} P\left(\frac1K \sum_{k=1}^K Z_{ik}^q < \frac12 \quad \forall i\right) d\mu(q) \\
& = 1 - \int_{q \in [0,1]} \left[P\left(\frac1K \sum_{k=1}^K Z_{1k}^q < \frac12\right)\right]^m d\mu(q) 
 \end{align*}
Define the function
\[f_K(q) = P\left(\frac1K \sum_{k=1}^K Z_{1k}^q < \frac12\right)\]
for $q \in [0,1]$. By the strong law of large numbers and the Central Limit Theorem, $f_K(q)$ has three possible limit points as $K$ grows large. This limit is 1 if $q\in [0,1/2)$, 0 if $q\in(1/2,1]$, and $1/2$ if $q=1/2$. Moreover, $0 \leq f_K(q) \leq g(q)$ pointwise, where the function $g(q)=1$ is integrable on $[0,1]$. So by the dominated convergence theorem
\begin{align*}
1 - \lim_{K \rightarrow \infty} &\int_{q \in [0,1]} \left[f_K(q)\right]^m d\mu(q) \\ & = 1 - \int_{q \in [0,1]} \lim_{K \rightarrow \infty} \left[f_K(q)\right]^m d\mu(q) \\
& = 1-\left[\mu\left(\left\{q<\frac12\right\}\right) \cdot 1 + \mu\left(\left\{q>\frac12\right\}\right) \cdot 0 + \mu\left(\left\{q=\frac12\right\}\right) \cdot \left(\frac12\right)^m \right] \\
& \leq \mu\left(\left\{q\geq \frac12\right\}\right) \\
& \leq 2\mathbb{E}_\mu(q) = \alpha
\end{align*} 
where the final line uses Markov's inequality and $\mathbb{E}_\mu(q)=\alpha/2$. Thus for every $\varepsilon>0$ there exists $\overline{K}_m$ such that 
\[ 1 - \int_{q \in [0,1]} \left[P\left(\frac1K \sum_{k=1}^K Z_{1k}^q < \frac12\right)\right]^m d\mu(q) < \alpha +\varepsilon\]
for all $K > \overline{K}_m$. Together with (\ref{eq:Complement}) this implies
\begin{equation} \label{eq:SufficientK}
P\left(\max_{1\leq i \leq m } \mathbbm{1}\left(\frac1K \sum_{k=1}^K \phi_{ik}^{\alpha/2}\geq \frac12 \right)=1\right) < \alpha +\varepsilon
\end{equation}
for $K \geq \overline{K}_m$.

We will now demonstrate a bound on $K$ that holds uniformly across the relevant values of $m$. Without loss, normalize $c_2(1)=0$.\footnote{For arbitrary cost functions $c_1$ and $c_2$, define $\tilde{c}_i(m) = c_i(m) - c_2(1)$ for each $i \in \{1,2\}$, noting that $\tilde{c}_2(1)=0$. The analyst's solutions $m_1^*$ and $m_2^*$ are unchanged by adding the constant $c_2(1)$ to his payoffs, and so the policymaker's payoffs are identical given cost functions $c_i$ or $\tilde{c}_i$.} Since $c_2(m)$ is (by assumption) weakly convex and strictly increasing, there exists an $\overline{m}$ such that $c_2(\overline{m})>\alpha + \varepsilon$. Define
 $\overline{K} = \max\left\{\overline{K}_1, \dots, \overline{K}_{\overline{m}}\right\}$. Then when $K > \overline{K}$,
\begin{equation} \label{eq:UniformBound}
v(m) \equiv P\left(\max_{1\leq i \leq m } \mathbbm{1}\left(\frac1K \sum_{k=1}^K \phi_{ik}^{\alpha/2}\geq \frac12 \right)=1\right) < \alpha +\varepsilon
\end{equation}
holds uniformly across $1 \leq m \leq \overline{m}$ by our previous arguments.

It remains to show that $m_2^* < \overline{m}$. Because $\int f_K(q)^m d\mu(q)$ is a positive mixture of convex functions, it is itself convex in $m$. Thus $v(m) = 1 - \int f_K(q)^m d\mu(q)$ is concave, and so the analyst's payoff
\[u(m) = v(m) - c_2(m)\]
is concave in $m$.  Moreover $u(1) =v(1) - c_2(1)> 0$ while $u(\overline{m}) < 0 $, since 
\begin{align*} 
v(\overline{m}) & < \alpha + \varepsilon  && \text{by  (\ref{eq:UniformBound})} \\
& < c_2(\overline{m}) && \text{by definition of $\overline{m}$}
\end{align*}
Thus there exists an $m' \leq \overline{m}$ such that $u(m) < 0$ for all $m\geq m'$. (See Figure \ref{fig:nbar} for a depiction of these points.) No $m\geq m'$ can be optimal, since the analyst could profitably deviate  to $m=1$. Thus the analyst's solution satisfies $m^*_2 < m' \leq \overline{m}$ as desired.

 \begin{figure}[h]
 \begin{center}
     \includegraphics[scale=0.32]{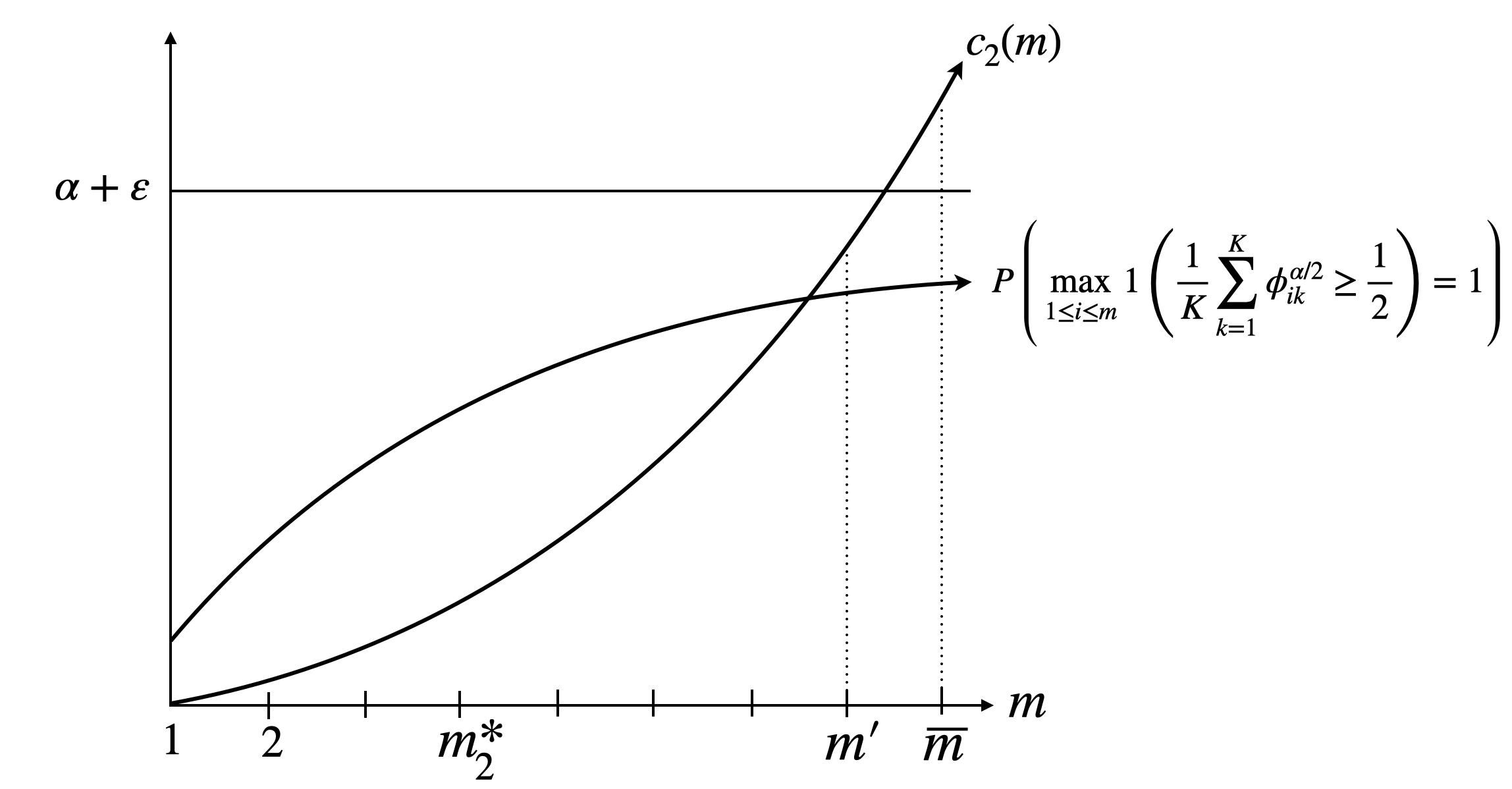}
     \caption{The analyst's solution satisfies $m_2^* \leq \overline{m}$ for some finite $\overline{m}$.} \label{fig:nbar}
     \end{center}
 \end{figure}
 
 \end{proof}

  \begin{lemma}  \label{lemm:Part1} For every $m > 1$, \begin{equation} \label{eq:Inequality1}
 P\left(\max_{1\leq i \leq m } \phi_{i1}^\alpha = 1\right)  > \alpha.
 \end{equation}
 Moreover $m_1^* > 1$.
  \end{lemma}

  \begin{proof}
Let $Z^q_1, Z^q_2, \dots$ be i.i.d Ber$(q)$ variables where $q \sim \mu$ and $\mathbb{E}_\mu(q)=\alpha$. Following identical arguments to those used in the proof of Lemma \ref{lemm:Part2}, we have
 \begin{align*}
 P\left(\max_{1\leq i \leq m} \phi_{i1}^\alpha = 1\right)  & = 1 - P\left(\phi_{i1}^{\alpha} =0 \,\, \forall i\right) \\
 & = 1 - \int_{q \in [0,1]} \left[P\left(Z^q_{1} =0\right)\right]^m d\mu(q)  \\
  & = 1 - \int_{q \in [0,1]} (1-q)^m d\mu(q) \\
  & \geq 1 - \int_{q \in [0,1]} (1-q) d\mu(q) && \text{since $q \in [0,1]$} \\
  & = 1-(1-\mathbb{E}_\mu(q)) = \alpha
 \end{align*}
When $m>1$, we can improve the weak inequality to a strict inequality by the following argument. Suppose instead that $\int_{q \in [0,1]} (1-q)^m d\mu(q) = \int_{q \in [0,1]} (1-q) d\mu(q)$. This (and the constraint $\mathbb{E}_\mu(q)=\alpha$) requires $\mu$ to place probability $1-\alpha$ on $q=0$ and probability $\alpha$ on $q=1$. But then $P(\phi_{21}^\alpha =1 \mid \phi_{11}^\alpha =0)=0$, contrary to Assumption \ref{assp:LowCost}.

So (\ref{eq:Inequality1}) holds for every $m>1$, as desired. It remains to argue that $m_1^*>1$. A sufficient condition is
\[P(\max \{\phi_{11}^\alpha,\phi_{21}^\alpha\}=1) - P(\phi_{11}^\alpha = 1) > c_1(2) - c_1(1) \]
i.e., the marginal value of the second repetition exceeds its cost. This reduces to $P(\phi_{21}^\alpha = 1 \mid \phi_{11}^\alpha = 0) \cdot P(\phi_{11}^\alpha =0) > c_1(2) - c_1(1)$, or
\[ P(\phi_{21}^\alpha = 1 \mid \phi_{11}^\alpha = 0) > \frac{c_1(2)-c_1(1)}{1-\alpha}\]
which is the content of Assumption \ref{assp:LowCost}. 
\end{proof}

Lemmas \ref{lemm:Part2} and \ref{lemm:Part1} yield
\[P\left(\max_{1\leq i \leq m_1^*} \phi_{i1}^\alpha = 1\right)  > \alpha + \varepsilon > P\left(\max_{1\leq i \leq m_2^*} \mathbbm{1}\left(\frac1K \sum_{k=1}^K \phi_{ik}^{\alpha/2}\geq \frac12 \right)=1\right)\]
for all $K > \overline{K}$, which further implies
 \[ u_P(s_1) = 1 - P\left(\max_{1\leq i \leq m_1^* } \phi_{i1}^\alpha = 1\right)  <  1 - P\left(\max_{1\leq i \leq m_2^* } \mathbbm{1}\left(\frac1K \sum_{k=1}^K \phi_{ik}^{\alpha/2}\geq \frac12 \right)=1\right) = u_P(s_2)\]
 for $K$ sufficiently large. This concludes the proof.

\subsection{Proof of Proposition \ref{prop:RobustManipulation}}

The proposition follows from the subsequent lemma.
\begin{lemma} \label{lemm:CompareValue} For every $K > -\frac{2\ln \alpha}{(1-\alpha)^2}$, 
\[P\left(\max_{1\leq i \leq m } \mathbbm{1}\left(\frac1K \sum_{k=1}^K \phi_{ik}^{\alpha/2} \geq \frac12 \right)=1\right)  < \alpha \leq P\left(\max_{1\leq i \leq m } \phi_{i1}^{\alpha} = 1\right) \quad \forall m \in \mathbb{Z}_+\]
\end{lemma}

 \begin{proof} By assumption that the infinite array $(p_i^k)$ has i.i.d entries, 
\begin{equation} \label{eq:Simplify1}
P\left(\max_{1\leq i \leq m } \phi_{i1}^\alpha = 1\right) = 1-\left[P\left(\phi_{i1}^\alpha =0\right)\right]^m = 1- (1-\alpha)^m \geq \alpha.
\end{equation}
while
 \begin{equation} \label{eq:Simplify2}
 P\left(\max_{1\leq i \leq m } \mathbbm{1}\left(\frac1K \sum_{k=1}^K \phi_{ik}^{\alpha/2} \geq \frac12 \right)=1\right)   = 1- \left[P\left(\frac1K \sum_{k=1}^K \phi_{1k}^{\alpha/2} < \frac12\right)\right]^m
 \end{equation}
So it is sufficient to show $P\left(\frac1K \sum_{k=1}^K \phi_{1k}^{\alpha/2} < \frac12\right)> 1-\alpha$, or equivalently
 \begin{equation} \label{eq:ToShow}
P\left(\frac1K \sum_{k=1}^K \phi_{1k}^{\alpha/2} \geq \frac12\right) < \alpha.
 \end{equation}
This inequality says that the probability of (incorrectly) rejecting under statistical test $s_1$ is higher than under statistical test $s_2$. To show this, define $S_K = \sum_{k=1}^K \phi_{1k}^{\alpha/2}$, observing that $S_K \sim \text{Binomial}(K,\alpha/2)$. Then by Hoeffding's inequality,
\[P(S_K \geq K/2) = P\left(S_K - \mathbb{E}(S_K) \geq \frac{K}{2}\left(1 - \alpha\right)\right) \leq e^{-\frac{K\left(1 - \alpha\right)^2}{2}}\]
So the desired (\ref{eq:ToShow}) is satisfied whenever $e^{-\frac{K\left(1 - \alpha\right)^2}{2}} < \alpha$, or equivalently, $K > -\frac{2\ln \alpha}{(1-\alpha)^2}$.
 \end{proof}

Lemma \ref{lemm:CompareValue} implies
\[
1-P\left(\max_{1\leq i \leq m_1^* } \phi_{i1}^\alpha = 1\right) \leq 1-\alpha < 1- P\left(\max_{1\leq i \leq m_2^*} \mathbbm{1}\left(\frac1K \sum_{k=1}^K \phi_{ik}^{\alpha/2}\geq \frac12 \right)=1\right)\]
and hence
 \[ u_P(s_1) = 1 - P\left(\max_{1\leq i \leq m_1^* } \phi_{i1}^\alpha = 1\right)  <  1 - P\left(\max_{1\leq i \leq m_2^* } \mathbbm{1}\left(\frac1K \sum_{k=1}^K \phi_{ik}^{\alpha/2}\geq \frac12 \right)=1\right) = u_P(s_2)\]
 for all $K > \frac{2\ln \alpha}{(1-\alpha)^2}$, as desired.

 \section{Simulation Exercise}\label{sec:sims}
In this section we study the effect of employing arbitrarily fair candidate algorithms on the power of our bootstrap procedure for testing fairness. To focus on this specific feature of the test, we abstract away from the details of our procedure and explicitly generate values of the status-quo and candidate fairness measures using a data generating process given by
\[
\begin{bmatrix}
\Gamma_0 \\
\Gamma_1
\end{bmatrix}
\sim N
\left(
\begin{bmatrix}
1.52 \\
\eta
\end{bmatrix},
\begin{bmatrix}
10 & 6.85 \\
6.85 & 10.83
\end{bmatrix}
\right)~.
\]
Here, $\Gamma_0$ and $\Gamma_1$ represent random draws of the utility discrepancy between groups (i.e. realizations of $u^{F}_{tr}(Z_i,\theta) - u^{F}_{tb}(Z_i,\theta)$) for the status-quo and candidate algorithms, respectively. Our null hypothesis of interest could thus be restated as
\[H_0: \left|E[\Gamma_1]\right| - |E[\Gamma_0]| \ge 0~,\]
against the alternative
\[H_1: |E[\Gamma_1]| -  |E[\Gamma_0]| < 0 ~.\]
The mean of $\Gamma_0$ and the covariance matrix for $[\Gamma_0, \Gamma_1]$ were calibrated using the test sample from Iteration 3 in Table \ref{tab:p_val}.\footnote{For random forests, both Iteration 2 and 3 yield the median $p$-value. To break the tie, we select Iteration 3 as the median split, as it produces a higher $p$-value for linear regression and LASSO (see also Table~\ref{tab:p_val_lasso} and \ref{tab:p_val}.)}
Our objective is to study the power of the test $\phi_n^{(f)}$ for a range of values of the mean fairness of the candidate algorithm given by $\eta \in [0, 1.75]$. For each $\eta$, we report the rejection probability of a $5\%$ level test based on $2,000$ bootstrap draws, computed using $5,000$ Monte Carlo iterations. To simplify the computation, we consider much smaller test set sample sizes  ($\ell_n \in \{100, 200, 400\}$) than the sample size available in the empirical application, but we simultaneously reduce the amount of noise in the problem considerably, by re-normalizing the covariance matrix so that $\text{Var}(\Gamma_0) = 10$. 

Figure \ref{fig:Bootpower} presents the computed power curves for each sample size. In all cases, we see that the test controls size when $\eta \ge 1.52$, as expected. As candidate fairness approaches zero, we see that power generally increases. However, when $\ell_n \in \{100, 200\}$, we also observe that the power of the test dips slightly as we approach perfect fairness. We conjecture that this may be related to the fact that the critical value generated by our bootstrap procedure could be less well-behaved at points near perfect fairness, as mentioned in the discussion following Theorem \ref{thm:consistent_test}. However, as expected given our theoretical results, this feature of the power curve disappears once the test sample is sufficiently large (in this case, $\ell_n = 400$).

\begin{figure}[htbp]
    \centering
    \begin{subfigure}[b]{0.45\textwidth}
        \centering
        \includegraphics[scale=0.45]{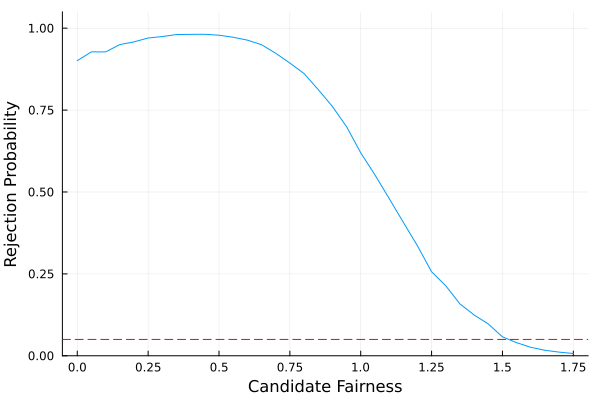}
        \caption{Sample size $\ell_n=100$}
        \label{fig:sub1}
    \end{subfigure} \\[1ex] 
    \begin{subfigure}[b]{0.45\textwidth}
        \centering
        \includegraphics[scale=0.45]{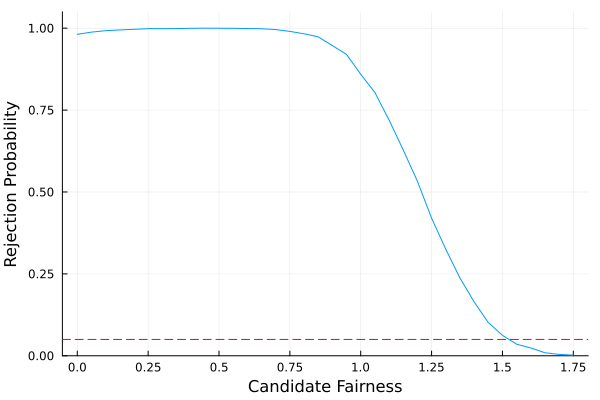}
        \caption{Sample size $\ell_n=200$}
        \label{fig:sub2}
    \end{subfigure} \\[1ex] 
    \begin{subfigure}[b]{0.45\textwidth}
        \centering
        \includegraphics[scale=0.45]{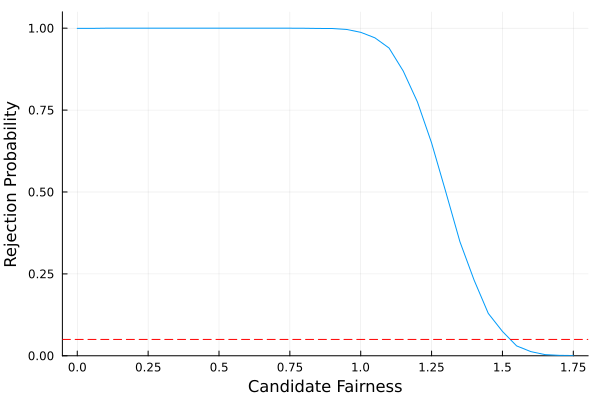}
        \caption{Sample size $\ell_n=400$}
        \label{fig:sub3}
    \end{subfigure}
    \caption{\footnotesize{Rejection probabilities for various levels of candidate fairness, for sample sizes $\ell_n \in \{100, 200, 400\}$.}}
    \label{fig:Bootpower}
\end{figure}

\section{Search for an Improving Algorithm} \label{sec:Search}
In this section we present a collection of systematic methods for finding a candidate algorithm which may improve over the status quo for two specific utility specifications: the classification rate utility (Example \ref{ex:Classification}) and the calibration utility (Example \ref{ex:Calibration}). Before proceeding, we note that our ultimate test for fairness/accuracy-improvability is valid regardless of the method used to construct the candidate algorithm, and good examples of possible methods are already available in the literature for some utility specifications (see Section~\ref{subsec:RelatedLit}.) With this in mind, our goal in this section is to provide methods which are specifically tailored to the problem of interest.

As in Section \ref{sec:af_improve}, let $a_0$ denote a status quo algorithm and let $\mathcal{A}$ denote some fixed class of algorithms. Our objective is to find a candidate algorithm $a_1 \in \mathcal{A}$ with which we can test for a fairness and/or accuracy improvement relative to $a_0$, using the test outlined in Section \ref{sec:Test}. For a given class of algorithms $\mathcal{A}$, let $\mathcal{U}(\mathcal{A})$ denote the feasible set of expected accuracy and fairness utilities that can be realized by algorithms in $\mathcal{A}$. Consider first the search for an algorithm that $\delta$-fairness improves on $a_0$. The following optimization problem
\begin{equation} \label{eq:optprob1}
\begin{aligned}
\inf_{U^r_F,U^b_F, U^r_A, U^b_A\in \mathbb{R}_+} \quad & |U^r_F - U^b_F|  \\
\textrm{s.t.} \quad & U^r_A > U^r_A(a_0)\\
  &   U^b_A > U^b_A(a_0) \\
  & (U^r_F,U^b_F, U^r_A, U^b_A) \in \mathcal{U}(\mathcal{A})~
\end{aligned}
\end{equation}
 returns the smallest possible unfairness, subject to feasibility and no reduction in accuracy. The status quo algorithm is $\delta$-fairness improvable if and only if the optimal value of this optimization problem is strictly less than $ |U^r_F(a_0) - U^b_F(a_0)|(1  -\delta)$. We could alternatively search for an algorithm that $\delta$-improves on accuracy for group (say, for instance) $r$ in the same way. Such an algorithm exists if and only if the optimal value of the following optimization problem is strictly greater than $U^r_A(a_0)(1 -\delta)$:
\begin{equation}\label{eq:optprob2}
\begin{aligned}
\sup_{U^r_F,U^b_F, U^r_A, U^b_A \in \mathbb{R}_+} \quad &  U^r_A\\
\textrm{s.t.} \quad & U^b_A > U^b_A(a_0) \\
  &  |U^r_F - U^b_F| < |U^r_F(a_0) - U^b_F(a_0)| \\
  & (U^r_F,U^b_F, U^r_A, U^b_A) \in \mathcal{U}(\mathcal{A})~.
\end{aligned}
\end{equation}

We propose to search for candidates $a_1 \in \mathcal{A}$ by solving feasible versions of the optimization problems \eqref{eq:optprob1}, \eqref{eq:optprob2} given a set of training data $\{(X_j, Y_j, G_j)\}_{j = 1}^m$. Let $\hat{U}^g_h(a)$ denote the empirical utility of algorithm $a$ for group $g \in \{r,b\}$, $h \in \{A, F\}$ computed using the sample $\{(X_j, Y_j, G_j)\}_{j = 1}^m$ (for the status quo algorithm $a_0$ we will often denote this as $\hat{U}^g_{0h}$). Let 
\[\widehat{\mathcal{U}}(\mathcal{A}) \equiv \{(\hat{U}^r_F(a),\hat{U}^b_F(a), \hat{U}^r_A(a), \hat{U}^b_A(a)) : a \in \mathcal{A}\}~,\]
be a sample analog of the feasible set.  A candidate for improvement could now be computed by solving the sample analog of optimization problem \eqref{eq:optprob1} given by

\begin{equation} \label{eq:optprob1_feas}
\begin{aligned}
\min_{U^r_F,U^b_F, U^r_A, U^b_A\in \mathbb{R}_+} \quad & |U^r_F - U^b_F|  \\
\textrm{s.t.} \quad & U^r_A \ge \hat{U}^r_{0A} + \iota\\
  &   U^b_A \ge \hat{U}^b_{0A} + \iota\\
  & (U^r_F,U^b_F, U^r_A, U^b_A) \in \widehat{\mathcal{U}}(\mathcal{A})~
\end{aligned}
\end{equation}

\noindent where $\iota \ge 0$ is some small positive constant. Solving \eqref{eq:optprob1_feas} may not be tractable in general. In the rest of this section we consider tractable versions of problem \eqref{eq:optprob1_feas} for specific utility specifications, by reformulating \eqref{eq:optprob1_feas} as a mixed integer linear program (MILP). Modern optimization libraries are capable of solving MILPs of moderate size \citep[recent applications in econometrics include][]{florios2008exact, kitagawa2018should,chen2018best,mbakop2021model}.

\subsection{Classification Utility}
Consider a setting with $\mathcal{Y} = \mathcal{D} = \{0, 1\}$, $w_A = w_F = w(x,y,d) = 1$, and $u_A = u_F = u(x, y, d) = \mathbbm{1}(y = d)$, so that our objective is to find a \emph{classifier} $a : \mathcal{X} \rightarrow \{0, 1\}$. Moreover, we fix $\mathcal{A}$ to be the set of linear classifiers:
\begin{definition} Algorithm $a: \mathbb{\mathcal{X}} \rightarrow \{0,1\}$ is a \emph{linear classifier} if $a(x) = \mathbbm{1}(\beta^T x \geq 0)$ for some $\beta \in \mathbb{R}^d$. From now on we use $\mathcal{A}$ to denote the set of all linear classifiers. (We assume that each vector $X_i$ contains a $1$ in its first entry, so it is without loss to set the threshold to zero.)
\end{definition} 
Given these restrictions, we can operationalize \eqref{eq:optprob1_feas} via the following mixed integer linear program (MILP) formulation of the feasible set $\widehat{\mathcal{U}}(\mathcal{A})$ (note here that we set the fairness and accuracy utilities to be equal, so that the feasible set consists of utility \emph{pairs}):

\begin{equation}\label{eq:e_hat}
\widehat{\mathcal{U}}(\mathcal{A}) = 
  \left\{(U^r, U^{b}) \in [0, 1]^2 \ 
    \begin{tabular}{|l}
      $U^r = \frac{1}{n_r}\sum_{G_j = r}\left((1 - Y_j) + (2Y_j - 1)D_j\right)$ \\[4pt]
      $U^b = \frac{1}{n_b}\sum_{G_j = b}\left((1 - Y_j) + (2Y_j - 1)D_j\right)$ \\[4pt]
      $\frac{X_j^T\beta}{C_j} < D_j \le 1 + \frac{X_j^T\beta}{C_j} \,\, \forall \hspace{1mm} 1 \le j \le m, \text{for some $\beta \in \mathcal{B}$}$ \\[4pt]
      $D_j\in \{0, 1\} \,\, \forall \hspace{1mm} 1 \le j \le m$
\end{tabular}
\right\}
\end{equation}
where $n_g = \sum_{G_j = g} 1$, and $C_j$ is chosen to satisfy $C_j > \sup_{\beta \in \mathcal B} |X_j^T \beta|$, with $\mathcal B$ some compact set such that we restrict $\beta \in \mathcal B$. The first, second and fourth conditions guarantee that  there is binary vector of decisions $(D_j)_{j = 1}^m$ (where $D_j = 1$ corresponds to classifying individual $j$ as $1$) yielding sample utilities $U^r$ and $U^b$ under classification loss. The inequalities in the third condition further constrain that these decisions are consistent with a linear classifier, since if $X_j^T\beta$ is weakly positive then the constraint  $D_j\in \{0,1\}$ implies $D_j =1$, while if $X_j^T\beta$ is strictly negative then the constraint $D_j \in \{0,1\}$ implies $D_j=0$. We can then use this characterization of $\widehat{\mathcal{U}}(\mathcal{A})$ to construct the following MILP formulation of \eqref{eq:optprob1_feas}: 

\begin{equation}\label{eq:find_fair}
\begin{aligned}
\min_{(t, U^r , U^b) \in [0, 1]^3} \quad  & t \\
\textrm{s.t.} \quad & t \ge U^r - U^b \\
& t \ge -(U^r - U^b) \\
& U^r \ge \hat{U}^r_{0} + \iota \\
& U^b \ge \hat{U}^b_{0} + \iota \\
& (U^r, U^b) \in \widehat{\mathcal{U}}(\mathcal{A})~,
\end{aligned}
\end{equation}
where we note that the first two constraints reformulate the absolute value in the objective function as a set of linear inequalities. We could also operationalize a feasible version of problem \eqref{eq:optprob2} as a MILP as well:
\begin{equation}\label{eq:find_pareto}
\begin{aligned}
\max_{(U^r, U^b) \in [0, 1]^2} \quad & U^r \\
\textrm{s.t.} \quad & U^b \ge \hat{U}^b_{0} + \iota\\ 
  &U^r - U^b \le |\hat{U}^r_{0} - \hat{U}^b_{0}| - \iota  \\
  &U^r - U^b \ge -|\hat{U}^r_{0} - \hat{U}^b_{0}| + \iota \\
  & (U^r, U^b) \in \widehat{\mathcal{U}}(\mathcal{A})~.
\end{aligned}
\end{equation}

Let $\beta_f$ denote a coefficient vector that achieves the minimum value of problem \eqref{eq:find_fair}. Then a candidate algorithm for $\delta$-fairness improvement is
$a_f(x) = \mathbbm{1}(\beta_f^T  x \geq 0)$. Let $\beta_r$ denote a coefficient vector that achieves the minimum value of of problem \eqref{eq:find_pareto}. Then a candidate algorithm for an $\delta$ -accuracy  improvement for group $r$ is $a_r(x) = \mathbbm{1}(\beta_r^T   x \geq 0)$.

\subsection{Calibration Utility}
Consider a setting in Example~\ref{ex:Calibration}: we have $\mathcal{Y} \subseteq \mathbb{R}$, $\mathcal{D} = \{0,1\}$,
\[
u(x,y,d) = y \mathbbm{1}(d=1), \quad w(x,y,d)=\mathbbm{1}(d=1), \quad u_A=u_F=u, \quad w_A=w_F=w. 
\]
Our objective is to find a linear classifier. As in our empirical application, we set the capacity constraint by restricting to algorithms that select a fraction $\kappa$ of the population.
With a slight abuse of notation, let $G_i \equiv 1$ when $G_i=b$; otherwise let $G_i \equiv 0$.
The error pair $(U^r, U^b) \in \mathbb{R}_+^2$ is in the feasible set $\hat{\mathcal{U}}(\mathcal{A})$ if and only if there exist $\beta \in \mathcal{B}$ such that
\begin{align}
\label{eq:feasible_calibration}
\left\{
\begin{matrix}
    \displaystyle{U^b = \left(\sum_i G_i D_i \right)^{-1} \sum_{i} Y_i G_i D_i} \\[8pt]
    \displaystyle{U^r = \left(\sum_i (1-G_i) D_i \right)^{-1} \sum_{i} Y_i (1-G_i) D_i}\\[8pt]
    \displaystyle{\sum_i D_i = \kappa m} \\[8pt]
    \frac{X_i^\top \beta}{C_i} < D_i \leq 1 + \frac{X_i^\top \beta}{C_i} 
    \ \forall i\\[8pt]
    D_i \in \{0,1\} \ \forall i
\end{matrix}
\right.
\end{align}
where $m$ in the third condition denotes the size of the train set.

We can formulate mixed integer programming (MIP) corresponding to \eqref{eq:find_fair} and \eqref{eq:find_pareto} for the current $\hat{\mathcal{U}}(\mathcal{A})$. However, the optimization problem is not linear in decision variables since the first and second conditions in \eqref{eq:feasible_calibration} are nonlinear. Below, we show that we could operationalize the problem by transforming it into a MILP.

To describe the transformation, let us introduce some variables: 
\[
    z_{1i} \equiv Y_i G_i, \
    z_{0i} \equiv Y_i (1 - G_i),\
    z_1 \equiv (z_{1i})_i, \
    z_0 \equiv (z_{0i})_i, \
    d \equiv (D_i)_i, \
    g \equiv (G_i)_i.
\]
Also, let $\mathbbm{1}$ denote the vector whose elements are all $1$. Then MIP corresponding \eqref{eq:find_fair}, say $(P)$, can be written as follows:

\[
\mathrm{(P)}
\left[
\begin{matrix}
    \displaystyle{\min_{\beta, d, U^b, U^r}} & \displaystyle{\left| U^b - U^r \right|} \\[8pt]
    \mathrm{s.t.}
    & U^b = \displaystyle{\frac{z_1^\top d}{g^\top d}} \\[8pt]
    & U^r = \displaystyle{\frac{z_0^\top d}{(1 - g)^\top d}} \\[8pt]
    & U^b \geq \hat U^b_0 + \iota \\
    & U^r \geq \hat U^r_0 + \iota \\
    & \displaystyle{\1^\top d = \kappa m} \\
    & \displaystyle{\frac{X_i^\top \beta}{C_i}} < D_i \leq 1 + \displaystyle{\frac{X_i^\top \beta}{C_i}}  \\[8pt]
    & D_i \in \{0,1\}, \beta \in \mathcal{B}
\end{matrix}
\right. 
\]
Note that for each group there exists at least one $i$ with $D_i=1$ at the optimal solution of $(P)$; otherwise the value of the objective function is $\infty$.
Consider the following change of variables:
\begin{align}
    t_1 \equiv \frac{1}{g^\top d}, \quad
    t_0 \equiv \frac{1}{(1-g)^\top d}, \quad
    s_1 \equiv t_1 d,\quad
    s_0 \equiv t_0 d.
        \label{eq:new_vars_MILP}
\end{align}
Then we can formulate the following MILP that corresponds to the original problem: 
\[
\mathrm{(P')}
\left[
\begin{matrix}
    \displaystyle{\min_{
    \substack{\beta, d, U^b, U^r, \\
    s_{1}, s_{0}, t_1, t_0
    }}} & \displaystyle{\left|U^b - U^r \right|} \\[8pt]
    \mathrm{s.t.}
    & U^b = z_1^\top s_1 \\
    & U^r = z_0^\top s_0 \\
    & U^b \geq \hat U^b_0 + \iota \\
    & U^r \geq \hat U^r_0 + \iota \\
    & \displaystyle{\1^\top d = \kappa m} \\
    & g^\top s_1 = 1 \\
    & (1-g)^\top s_0 = 1 \\[8pt]
    & \displaystyle{\frac{X_i^\top \beta}{C_i}} < D_{i} \leq \displaystyle{\left( 1 + \frac{X_i^\top \beta}{C_i}\right)} \ \forall i \\[8pt]
    & s_{1i} \leq t_1 \ \forall i \\
    & s_{1i} \leq D_{i} \ \forall i \\
    & s_{1i} \geq t_1 - (1 - D_{i})  \ \forall i \\
    & s_{0i} \leq t_0  \ \forall i \\
    & s_{0i} \leq D_{i}  \ \forall i \\
    & s_{0i} \geq t_0 - (1 - D_{i})  \ \forall i \\
    & t_1 \in (0,1], t_0 \in (0,1], s_{1i} \geq 0, s_{0i} \geq 0, U^b \geq 0, U^r \geq 0, \\
    & D_i \in \{0,1\}, \beta \in \mathcal{B}
\end{matrix}
\right.
\]

The following proposition guarantees that we can utilize the MILP $(P')$ to find a candidate algorithm for calibration utilities. Therefore, if the problem size is moderate, we can find a candidate algorithm using modern optimization libraries.

\begin{proposition}
\label{prop:milp_calibration}
    The optimal values of $(P)$ and $(P')$ coincide. Moreover, the optimal solution to $(P)$ can be computed by using the optimal solution to $(P')$.
\end{proposition}
\begin{proof}
    Denote the optimal values of $(P)$ and $(P')$ by OPT$(P)$ and OPT$(P')$, respectively.
    For each solution feasible in $(P)$, there is a corresponding solution feasible in $(P')$.
    Specifically, given $(\beta, d, U^r, U^b)$ feasible in $(P)$, $(\beta, d, U^r, U^b, s_1, s_0, t_1, t_0)$, where $(s_1, s_0, t_1, t_0)$ are defined in \eqref{eq:new_vars_MILP}, is feasible in $(P')$ and gives the same value of the objective funtion.
    Thus, OPT$(P')$ $\leq$ OPT$(P)$.
    
    To show OPT$(P')$ $\geq$ OPT$(P)$, we show that for each feasible solution in $(P')$, there exists a corresponding solution feasible in $(P)$ with the same value of the objective function.
    It suffices to show that
    \[
    \frac{s_{1i}}{t_1} = \frac{s_{0i}}{t_0} = D_i,
    \]
    for each $i$
    if $(\beta, d, U^r, U^b, s_1, s_0, t_1, t_0)$ is feasible in $(P')$, which guarantees that \eqref{eq:new_vars_MILP} indeed holds and $(\beta, d, U^r, U^b)$ is feasible in $(P)$.
    Below, we will show that
    the conditions
    \begin{align}
        s_{1i} &\leq t_1, \label{eq:milp_cond_1}\\
        s_{1i} &\leq D_i, \label{eq:milp_cond_2} \\
        s_{1i} &\geq t_1 - (1 - D_i), \label{eq:milp_cond_3} \\
        s_{1i} &\geq 0, \text{ and } \label{eq:milp_cond_4}\\
        t_1 &\in (0,1] \label{eq:milp_cond_5}, 
    \end{align}
    imply $s_{1i}/t_1 = D_i$ ($s_{0i}/t_0 = D_i$ can be shown in the same manner).
    Suppose that $D_i=1$. \eqref{eq:milp_cond_3} implies $s_{1i} \geq t_i$. Then, \eqref{eq:milp_cond_1} implies $s_{1i}=t_1$, and thus we have $s_{1i}/t_1 = 1 = D_i$.
    Next, suppose that $D_i=0$. \eqref{eq:milp_cond_2} and \eqref{eq:milp_cond_4} imply $s_{1i}=0$. By \eqref{eq:milp_cond_5}, we have $s_{1i}/t_1 = 0 = D_i$.
 \end{proof}

The problem corresponding to \eqref{eq:find_pareto} can be formulated in a similar manner.

\section{More detail about the legal framework}
\subsection{Title VII of the Civil Rights Act of 1964}
The exact legal process for evaluating a disparate impact case depends on the context and in particular the law prohibiting the behavior. Different laws provide different rights, remedies, and standards of evidence. However, there is a general framework that is commonly used as a blueprint. The framework was developed in the context of Title VII of the Civil Rights Act of 1964, which concerns employment discrimination.\footnote{Historically, the framework was introduced by the US Supreme Court in the case \cite{griggs1971}. The Civil Rights Act of 1964 was later amended by the Civil Rights Act of 1991 to codify the framework.} In particular, the law prohibits employers from using ``a facially neutral employment practice that has an unjustified adverse impact on members of a protected class. A facially neutral employment practice is one that does not appear to be discriminatory on its face; rather it is one that is discriminatory in its application or effect.’’

The process itself is outlined in 42 U.S. Code § 2000e–2(k) which reads
\begin{itemize}
\item[(k)] Burden of proof in disparate impact cases
\begin{itemize}
\item[(1)] (A) An unlawful employment practice based on disparate impact is established under this subchapter only if---
\begin{itemize}
\item[(i)] a complaining party demonstrates that a respondent uses a particular employment practice that causes a disparate impact on the basis of race, color, religion, sex, or national origin and the respondent fails to demonstrate that the challenged practice is job related for the position in question and consistent with business necessity; or

\item[(ii)] the complaining party makes the demonstration described in subparagraph (C) with respect to an alternative employment practice and the respondent refuses to adopt such alternative employment practice.
\end{itemize}

(B) \begin{itemize}
\item[(i)] With respect to demonstrating that a particular employment practice causes a disparate impact as described in subparagraph (A)(i), the complaining party shall demonstrate that each particular challenged employment practice causes a disparate impact, except that if the complaining party can demonstrate to the court that the elements of a respondent's decision making process are not capable of separation for analysis, the decision making process may be analyzed as one employment practice.

\item[(ii)] If the respondent demonstrates that a specific employment practice does not cause the disparate impact, the respondent shall not be required to demonstrate that such practice is required by business necessity.
\end{itemize}

(C) The demonstration referred to by subparagraph (A)(ii) shall be in accordance with the law as it existed on June 4, 1989, with respect to the concept of ``alternative employment practice".

\item[(2)] A demonstration that an employment practice is required by business necessity may not be used as a defense against a claim of intentional discrimination under this subchapter.

\item[(3)] Notwithstanding any other provision of this subchapter, a rule barring the employment of an individual who currently and knowingly uses or possesses a controlled substance, as defined in schedules I and II of section 102(6) of the Controlled Substances Act (21 U.S.C. 802(6)), other than the use or possession of a drug taken under the supervision of a licensed health care professional, or any other use or possession authorized by the Controlled Substances Act [21 U.S.C. 801 et seq.] or any other provision of Federal law, shall be considered an unlawful employment practice under this subchapter only if such rule is adopted or applied with an intent to discriminate because of race, color, religion, sex, or national origin.
\end{itemize}
\end{itemize}

This process was originally designed to be used by parties seeking judicial enforcement, and so is adversarial. For example, the complaining party may be an employee who is asking a judge to prohibit a potentially discriminatory employment practice and the respondent is the employer implementing the practice. In this case, the judge would evaluate the three parts of the process (the prima facie case, the business necessity defense, and the existence of suitable alternatives), and then determine whether the practice is prohibited under Title VII. 

\subsection{Title VI of the Civil Rights Act of 1964}

Many policies of interest to economists are implemented by organizations that receive federal funding. For example, many hospitals receive funding from Medicare or Hill-Burton. Because these organizations receive federal funding, their policies are prohibited from having a disparate impact by Title VI of the Civil Rights Act of 1964. This prohibition is described in 28 C.F.R. § 42.104(b)(2) which reads

\begin{itemize}
\item[(2)] A recipient, in determining the type of disposition, services, financial aid, benefits, or facilities which will be provided under any such program, or the class of individuals to whom, or the situations in which, such will be provided under any such program, or the class of individuals to be afforded an opportunity to participate in any such program, may not, directly or through contractual or other arrangements, utilize criteria or methods of administration which have the effect of subjecting individuals to discrimination because of their race, color, or national origin, or have the effect of defeating or substantially impairing accomplishment of the objectives of the program as respects individuals of a particular race, color, or national origin.
\end{itemize}

A key difference between that of a Title VII disparate impact case, is that in the setting of Title VI there is no private right to action \citep[see, specifically,][]{alexander2001}. That is, individuals cannot file civil actions relying on the Title VI disparate impact standard. Instead, it is the responsibility of the funding agency to enforce the law, which can take place both by investigating individual complaints and conducting compliance reviews. However, to fulfill this responsibility, the Department of Justice still recommends using the process outlined for Title VII described above. In particular, the Department of Justice Civil Rights Division’s Title VI Manual (available at \url{https://www.justice.gov/crt/fcs/T6manual}) writes:

``In administrative investigations, this court-developed burden shifting framework serves as a useful paradigm for organizing the evidence. Agency investigations, however, often follow a non-adversarial model in which the agency collects all relevant evidence then determines whether the evidence establishes discrimination. Under this model, agencies often do not shift the burdens between complainant and recipient when making findings. For agencies using this method, the following sections serve as a resource for conducting an investigation and developing an administrative enforcement action where appropriate.’’

In addition, while the Title VI Manual identifies the funding agency as ultimately  responsible for evaluating less discriminatory alternatives, it also writes that 
``Although agencies bear the burden of evaluating less discriminatory alternatives, agencies sometimes impose additional requirements on recipients to consider alternatives before taking action. These requirements can affect the legal framework by requiring recipients to develop the evidentiary record related to alternatives as a matter of course, before and regardless of whether an administrative complaint is even filed. Such requirements recognize that the recipient is in the best position to complete this task, having the best understanding of its goals, and far more ready access to the information necessary to identify alternatives and conduct a meaningful analysis.''

\section{Supplementary Material to Section~\ref{sec:Application}}

\subsection{Robustness checks related to selection rule choice}
\label{app:OtherTables}

We report here the point estimates of accuracy and fairness for both the candidate and status quo algorithms, as well as the $p$-values that emerge from our statistical test of the null hypothesis that the status quo algorithm is not strictly FA-dominated, for LASSO (Table~\ref{tab:p_val_lasso}) and linear regression (Table~\ref{tab:p_val}). The results are nearly identical to those reported in Table~\ref{tab:p_val_rf} for random forests.

\begin{table}[hbt]

    \footnotesize
    \begin{tabular}{
        @{}
        l
        *{3}{c}
        |
        *{3}{c}
        |
        *{3}{c}
        |
        c
        @{}
    }
    \toprule
    & \multicolumn{3}{c}{\bfseries Accuracy (Black)} & \multicolumn{3}{c}{\bfseries Accuracy (White)} & \multicolumn{3}{c}{\bfseries Unfairness} &  \\
    \cmidrule(lr){2-4} \cmidrule(lr){5-7} \cmidrule(lr){8-10} \cmidrule(l){11-11}
         & {$a_1$} & {$a_0$} & {$p_b$} & {$a_1$} & {$a_0$} & {$p_w$} & {$a_1$} & {$a_0$} & {$p_f$} & {$p$} \\
    \midrule
        Iteration 1 & 7.18 & 6.33 & 0.0005 & 7.27 & 5.14 & 0.0000 & 0.10 & 1.19 & 0.0000 & 0.0005 \\
        Iteration 2 & 7.39 & 6.32 & 0.0001 & 7.30 & 5.11 & 0.0000 & 0.09 & 1.20 & 0.0000 & 0.0001 \\
        Iteration 3 & 7.52 & 6.67 & 0.0002 & 7.11 & 5.15 & 0.0000 & 0.41 & 1.52 & 0.0000 & 0.0002 \\
        Iteration 4 & 7.20 & 6.35 & 0.0001 & 7.25 & 5.06 & 0.0000 & 0.05 & 1.28 & 0.0000 & 0.0001 \\
        Iteration 5 & 7.58 & 6.88 & 0.0040 & 7.27 & 5.27 & 0.0000 & 0.30 & 1.61 & 0.0000 & 0.0040 \\
        Iteration 6 & 7.97 & 6.52 & 0.0000 & 7.34 & 5.02 & 0.0000 & 0.63 & 1.51 & 0.0022 & 0.0022 \\
        Iteration 7 & 7.59 & 6.74 & 0.0009 & 7.34 & 5.19 & 0.0000 & 0.25 & 1.55 & 0.0000 & 0.0009 \\
    \bottomrule
    \end{tabular}
    
    \caption{\footnotesize{The candidate algorithm $a_1$ in the table is based on LASSO. Reported $p$-values are computed via bootstrap with 10,000 iterations. The median $p$-value $p_{\mathrm{med}}$ is 0.0005.}}
    \label{tab:p_val_lasso}

\end{table}

\begin{table}[hbt]

    \footnotesize
    \begin{tabular}{
        @{}
        l
        *{3}{c}
        |
        *{3}{c}
        |
        *{3}{c}
        |
        c
        @{}
    }
    \toprule
    & \multicolumn{3}{c}{\bfseries Accuracy (Black)} & \multicolumn{3}{c}{\bfseries Accuracy (White)} & \multicolumn{3}{c}{\bfseries Unfairness} &  \\
    \cmidrule(lr){2-4} \cmidrule(lr){5-7} \cmidrule(lr){8-10} \cmidrule(l){11-11}
         & {$a_1$} & {$a_0$} & {$p_b$} & {$a_1$} & {$a_0$} & {$p_w$} & {$a_1$} & {$a_0$} & {$p_f$} & {$p$} \\
    \midrule
        Iteration 1 & 7.31 & 6.33 & 0.0000 & 7.27 & 5.14 & 0.0000 & 0.04 & 1.19 & 0.0000 & 0.0000 \\
        Iteration 2 & 7.39 & 6.32 & 0.0001 & 7.23 & 5.11 & 0.0000 & 0.15 & 1.20 & 0.0000 & 0.0001 \\
        Iteration 3 & 7.49 & 6.67 & 0.0003 & 7.06 & 5.15 & 0.0000 & 0.43 & 1.52 & 0.0000 & 0.0003 \\
        Iteration 4 & 7.35 & 6.35 & 0.0000 & 7.27 & 5.06 & 0.0000 & 0.08 & 1.28 & 0.0000 & 0.0000 \\
        Iteration 5 & 7.58 & 6.88 & 0.0036 & 7.26 & 5.27 & 0.0000 & 0.32 & 1.61 & 0.0000 & 0.0036 \\
        Iteration 6 & 7.96 & 6.52 & 0.0000 & 7.26 & 5.02 & 0.0000 & 0.70 & 1.51 & 0.0049 & 0.0049 \\
        Iteration 7 & 7.65 & 6.74 & 0.0004 & 7.37 & 5.19 & 0.0000 & 0.28 & 1.55 & 0.0001 & 0.0004 \\
    \bottomrule
    \end{tabular}
    
    \caption{\footnotesize{The candidate algorithm $a_1$ in the table is based on linear regression. Reported $p$-values are computed via bootstrap with 10,000 iterations. The median $p$-value $p_{\mathrm{med}}$ is 0.0003.}}
    \label{tab:p_val}

\end{table}

\subsection{Robustness checks related to program effect}
\label{app:program_effect}
In our baseline analysis, we use the total number of active chronic illnesses in the year following enrollment in the management program as the outcome variable $Y$. As a robustness check, here we repeat this analysis using the total number of active chronic illnesses in the \emph{same} year as enrollment in the management program. We report point estimates of accuracy and fairness for both the candidate and status quo algorithms, along with corresponding $p$-values. The results are consistent with those observed in our baseline analysis: we reject the null hypothesis that it is not possible to reduce the algorithm's disparate impact without compromising on the accuracy of its predictions. For additional robustness checks related to the size of the program effect, see the supplementary materials of \cite{Obermeyer2019-te}.

\begin{figure}[h]
    \centering
    \includegraphics[height=9cm]{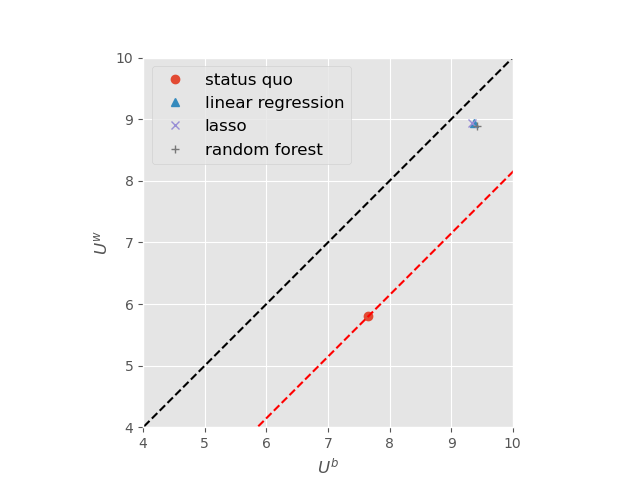}
   \caption{     \footnotesize{We report the average value of $U^g$ (i.e., number of active chronic diseases \emph{in the same year as the enrollment decision},  conditional on automatic enrollment) for each group $g$ across the $K=7$ iterations of our procedure. This figure is qualitatively similar to Figure \ref{fig:ML_improvement}.}}     \label{fig:ML_improvement_robustness}
\end{figure}

\begin{table}[hbt]

    \footnotesize 
    \begin{tabular}{
        @{} 
        l 
        *{3}{c} 
        | 
        *{3}{c} 
        |
        *{3}{c}
        |
        c 
        @{} 
    }
    \toprule
    & \multicolumn{3}{c}{\bfseries Accuracy (Black)} & \multicolumn{3}{c}{\bfseries Accuracy (White)} & \multicolumn{3}{c}{\bfseries Unfairness} &  \\
    \cmidrule(lr){2-4} \cmidrule(lr){5-7} \cmidrule(lr){8-10} \cmidrule(l){11-11}
         & {$a_1$} & {$a_0$} & {$p_b$} & {$a_1$} & {$a_0$} & {$p_w$} & {$a_1$} & {$a_0$} & {$p_f$} & {$p$} \\
    \midrule
        Iteration 1 & 9.35 & 7.43 & 0.0000 & 9.08 & 5.80 & 0.0000 & 0.27 & 1.62 & 0.0000 & 0.0000 \\
        Iteration 2 & 9.14 & 7.54 & 0.0000 & 8.91 & 5.85 & 0.0000 & 0.22 & 1.70 & 0.0000 & 0.0000 \\
        Iteration 3 & 9.27 & 7.82 & 0.0000 & 8.93 & 5.92 & 0.0000 & 0.34 & 1.90 & 0.0000 & 0.0000 \\
        Iteration 4 & 9.26 & 7.43 & 0.0000 & 8.83 & 5.67 & 0.0000 & 0.43 & 1.76 & 0.0000 & 0.0000 \\
        Iteration 5 & 9.40 & 8.14 & 0.0000 & 8.88 & 5.92 & 0.0000 & 0.53 & 2.22 & 0.0000 & 0.0000 \\
        Iteration 6 & 9.61 & 7.33 & 0.0000 & 9.01 & 5.66 & 0.0000 & 0.61 & 1.67 & 0.0003 & 0.0003 \\
        Iteration 7 & 9.59 & 7.91 & 0.0000 & 8.92 & 5.80 & 0.0000 & 0.67 & 2.11 & 0.0000 & 0.0000 \\
    \bottomrule
    \end{tabular}
    
    \caption{\footnotesize{The candidate algorithm $a_1$ in the table is based on linear regression. Reported $p$-values are computed via bootstrap with 10,000 iterations. The median $p$-value is smaller than 0.0001.}}
    \label{tab:p_val_linreg_robustness}

\end{table}

\begin{table}[hbt]

    \footnotesize 
    \begin{tabular}{
        @{} 
        l 
        *{3}{c} 
        | 
        *{3}{c} 
        |
        *{3}{c}
        |
        c 
        @{} 
    }
    \toprule
    & \multicolumn{3}{c}{\bfseries Accuracy (Black)} & \multicolumn{3}{c}{\bfseries Accuracy (White)} & \multicolumn{3}{c}{\bfseries Unfairness} &  \\
    \cmidrule(lr){2-4} \cmidrule(lr){5-7} \cmidrule(lr){8-10} \cmidrule(l){11-11}
         & {$a_1$} & {$a_0$} & {$p_b$} & {$a_1$} & {$a_0$} & {$p_w$} & {$a_1$} & {$a_0$} & {$p_f$} & {$p$} \\
    \midrule
        Iteration 1 & 9.31 & 7.43 & 0.0000 & 9.04 & 5.80 & 0.0000 & 0.27 & 1.62 & 0.0000 & 0.0000 \\
        Iteration 2 & 9.20 & 7.54 & 0.0000 & 8.92 & 5.85 & 0.0000 & 0.28 & 1.70 & 0.0000 & 0.0000 \\
        Iteration 3 & 9.25 & 7.82 & 0.0000 & 8.95 & 5.92 & 0.0000 & 0.30 & 1.90 & 0.0000 & 0.0000 \\
        Iteration 4 & 9.28 & 7.43 & 0.0000 & 8.84 & 5.67 & 0.0000 & 0.44 & 1.76 & 0.0000 & 0.0000 \\
        Iteration 5 & 9.43 & 8.14 & 0.0000 & 8.89 & 5.92 & 0.0000 & 0.54 & 2.22 & 0.0000 & 0.0000 \\
        Iteration 6 & 9.46 & 7.33 & 0.0000 & 8.96 & 5.66 & 0.0000 & 0.50 & 1.67 & 0.0001 & 0.0001 \\
        Iteration 7 & 9.44 & 7.91 & 0.0000 & 8.99 & 5.80 & 0.0000 & 0.45 & 2.11 & 0.0000 & 0.0000 \\
    \bottomrule
    \end{tabular}
    
    \caption{\footnotesize{The candidate algorithm $a_1$ in the table is based on LASSO. Reported $p$-values are computed via bootstrap with size 10,000. The median $p$-value is smaller than 0.0001.}}
    \label{tab:p_val_lasso_robustness}

\end{table}

\begin{table}[hbt]

    \footnotesize 
    \begin{tabular}{
        @{} 
        l 
        *{3}{c} 
        | 
        *{3}{c} 
        |
        *{3}{c}
        |
        c 
        @{} 
    }
    \toprule
    & \multicolumn{3}{c}{\bfseries Accuracy (Black)} & \multicolumn{3}{c}{\bfseries Accuracy (White)} & \multicolumn{3}{c}{\bfseries Unfairness} &  \\
    \cmidrule(lr){2-4} \cmidrule(lr){5-7} \cmidrule(lr){8-10} \cmidrule(l){11-11}
         & {$a_1$} & {$a_0$} & {$p_b$} & {$a_1$} & {$a_0$} & {$p_w$} & {$a_1$} & {$a_0$} & {$p_f$} & {$p$} \\
    \midrule
        Iteration 1 & 9.38 & 7.43 & 0.0000 & 8.91 & 5.80 & 0.0000 & 0.47 & 1.62 & 0.0004 & 0.0004 \\
        Iteration 2 & 9.31 & 7.54 & 0.0000 & 8.86 & 5.85 & 0.0000 & 0.45 & 1.70 & 0.0000 & 0.0000 \\
        Iteration 3 & 9.15 & 7.82 & 0.0000 & 8.90 & 5.92 & 0.0000 & 0.25 & 1.90 & 0.0000 & 0.0000 \\
        Iteration 4 & 9.31 & 7.43 & 0.0000 & 8.83 & 5.67 & 0.0000 & 0.48 & 1.76 & 0.0000 & 0.0000 \\
        Iteration 5 & 9.72 & 8.14 & 0.0000 & 8.94 & 5.92 & 0.0000 & 0.77 & 2.22 & 0.0000 & 0.0000 \\
        Iteration 6 & 9.51 & 7.33 & 0.0000 & 8.94 & 5.66 & 0.0000 & 0.57 & 1.67 & 0.0004 & 0.0004 \\
        Iteration 7 & 9.54 & 7.91 & 0.0000 & 8.89 & 5.80 & 0.0000 & 0.65 & 2.11 & 0.0000 & 0.0000 \\
    \bottomrule
    \end{tabular}
    
    \caption{\footnotesize{The candidate algorithm $a_1$ in the table is based on random forest. Reported $p$-values are computed via bootstrap with size 10,000. The median $p$-value is smaller than 0.0001.}}
    \label{tab:p_val_rf_robustness}

\end{table}

\subsection{Alternative Specification: Cost-Based Accuracy Measure}
\label{app:accuracy_cost}

Our baseline analysis uses identical utility functions to evaluate both accuracy and fairness. Here we demonstrate the flexibility of our framework by employing different utility functions for these two criteria. Specifically, while maintaining our original fairness measure based on expected health conditions, we modify the accuracy measure to reflect expected healthcare costs.

We implement this alternative specification using the same selection rule as in our baseline analysis, with the score assignment rule $\hat{f}$ trained to predict health conditions. It is worth noting that this creates an interesting tension: while our accuracy criterion now evaluates algorithms based on their ability to predict costs, our selection rule remains oriented toward predicting health conditions.

Figure~\ref{fig:2D_plot_cost_health_trainhealth} presents results from this modified analysis. In contrast to our baseline findings, we cannot reject the null hypothesis that the status quo algorithm is FA-dominated under these criteria. This result likely stems from two key factors. First, our selection rule is optimized for predicting health conditions rather than costs. Second, as documented by \cite{Obermeyer2019-te}, the status quo algorithm was specifically designed to predict healthcare costs. This suggests that when evaluated on its intended objective (cost prediction), the status quo algorithm performs relatively well, though potentially at the expense of other objectives, such as equitable prediction of health needs.

These findings should be interpreted with appropriate caution. The inability to reject the null may reflect limitations of our specific selection rule rather than a fundamental impossibility of improvement. Alternative selection rules, particularly those optimized for cost prediction, might yield different conclusions. This highlights the broader methodological point that selection rule choice can be consequential for testing FA-improvability.

\begin{figure}
    \centering
    \includegraphics[width=0.7\linewidth]{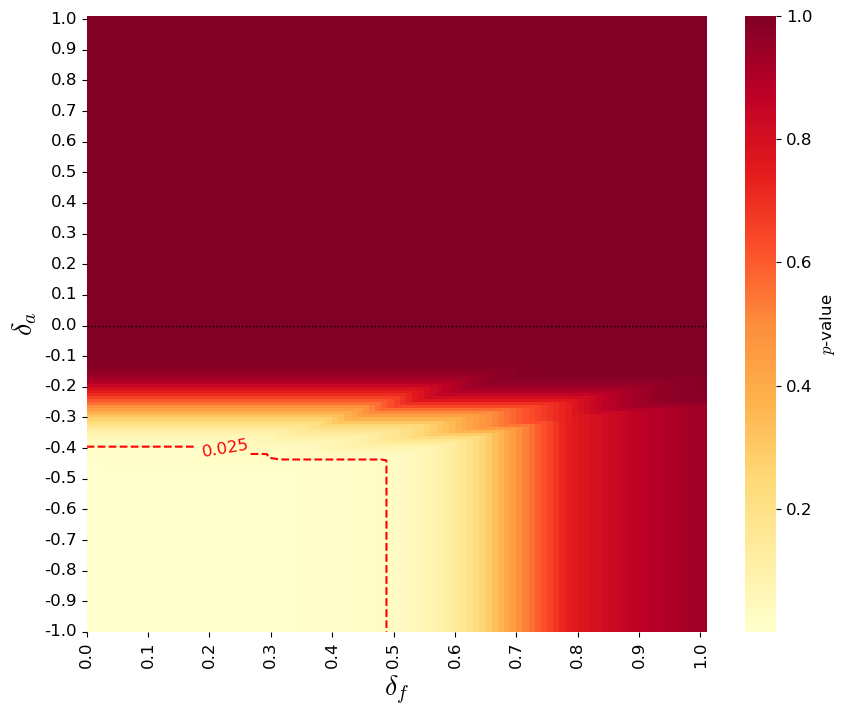}
    \caption{\footnotesize
    We present $p$-values for testing $(\delta_a, \delta_a, \delta_f)$-improvability for $(\delta_a, \delta_f) \in [-1,1] \times [0,1]$. The selection rule is based on random forests trained to predict expected health needs. As in the main text, we set $K = 7$ in our procedure. The fairness utility function measures the expected health needs of those selected into the program, while the accuracy utility functino measures the expected health costs of those selected. Larger values of $\delta_a$ and $\delta_f$ indicate a stricter test regarding the dimensions of fairness and accuracy, respectively. We find that the $p$-value exceeds $0.025$ at $(\delta_f, \delta_a) = (0,0)$, so we cannot reject the null that the status quo algorithm is not FA-improvable.}
    
    \label{fig:2D_plot_cost_health_trainhealth}
\end{figure}

\clearpage

\bibliographystyle{aer}
\bibliography{references}

\end{document}